\documentclass{article}
\usepackage{amsmath}
\usepackage{amssymb}
\usepackage{amsthm}
\usepackage{graphicx}
\usepackage[utf8]{inputenc}
\allowdisplaybreaks
\usepackage{caption}
\usepackage{cite}
\usepackage{authblk}
\usepackage{esint}
\usepackage{bm}

\newtheorem{theorem}{Theorem}
\newtheorem{lemma}{Lemma}
\newtheorem{corollary}{Corollary}
\newtheorem{remark}{Remark}

\usepackage{fullpage}

\usepackage{fullpage}

\title{Velocity of viscous fingers in miscible displacement:\\ Intermediate concentration}

\author[1]{Fedor Bakharev}
\author[1]{Aleksandr Enin}
\author[1]{Sergey Matveenko}
\author[1]{Dmitry Pavlov}
\author[2]{\\Yulia Petrova}
\author[1,3]{Nikita Rastegaev}
\author[1,2,4]{Sergey Tikhomirov}

\affil[1]{\small{St. Petersburg State University,
University Embankment, 7/9, St. Peterburg, 199034 Russia}}
\affil[2]{\small{Pontifícia Universidade Católica do Rio de Janeiro, Rua Marquês de São Vicente, 225,\newline Gávea Rio de Janeiro, RJ - Brasil Cep: 22451-900}}
\affil[3]{\small{St. Petersburg Department of V.~A.~Steklov Institute of Mathematics of the Russian Academy of Sciences, \newline
191023, 27 Fontanka, St. Petersburg, Russia
}}
\affil[4]{\small{Corresponding Author, sergey.tikhomirov@gmail.com
}}

\date{May 2024}

\begin{document}





\maketitle

\begin{abstract} 
We investigate one-phase flow in porous medium corresponding to a miscible displacement process in which the viscosity of the injected fluid is smaller than the viscosity in the reservoir fluid, which frequently leads to the formation of a mixing zone characterized by thin fingers. The mixing zone grows in time due to the difference in speed between its leading and trailing edges. The transverse flow equilibrium (TFE) model provides estimates of these speeds. We propose an enhancement for the TFE estimates, and provide its theoretical justification. It is based on the assumption that an intermediate concentration exists near the tip of the finger, which allows to reduce the integration interval in the speed estimate. Numerical simulations were conducted that corroborate the new estimates within the computational fluid dynamics model. The refined estimates offer greater accuracy than those provided by the original TFE model.  

\smallskip
\noindent {\bf Keywords. } viscous fingers, porous media, miscible displacement, transverse flow equilibrium.

\end{abstract}

\section{Introduction. Motivation}
Like many similar processes, the miscible displacement of a viscous liquid with a liquid of smaller viscosity in a porous medium is inherently unstable\footnote{In the paper the word ``viscosity'' refers to dynamic viscosity for liquids, see~\cite[Chap.~1,5]{bedrikovetsky2013} for details on role of dynamic viscosity for flows in porous media.}. The study of instabilities of this type originates in the pioneering works on the displacement in a Hele-Shaw cell by Saffman \& Taylor \cite{SaffmanTaylor1958} in the context of immiscible displacement and by Wooding \cite{Wooding1969} in the context of miscible displacement. The onset of instability usually leads to the formation of a mixing zone at the region of the viscosity jump. This zone is commonly composed of thin piercing fingers. For an overview of the early works on viscous fingering in porous media see \cite{Homsy1987,Tanveer2000,YYS2002}, and for more recent ones~\cite{Nijjer2018,scovazzi2017} (and references therein). Let us note that the mixing zone occurs despite the stabilizing effect of diffusion, see e.g. \cite{ChuokeMeurs1959,Outmans1962,Perrine1961-I,Perrine1961-II, Claridge}.

Viscous miscible displacement appears in petroleum engineering, especially in enhanced oil recovery methods \cite{Lake,green1998eor}, for example displacement of oil by a solvent \cite{Peaceman1962} and a polymer flooding  followed by a water post-flush \cite[Chap. 8]{Lake}, \cite[Chap. 5,7]{green1998eor}. 
The prototype example for this paper was the post-flush process of a polymer slug in surfactant-polymer flooding~\cite[Chap. 7]{green1998eor},~\cite{samanta2011}. 
While surfactant-polymer flooding is clearly a multiphase problem, we study the development of viscous fingers in the so-called post-flush zone---the region between the injection well and the polymer-surfactant slug; where there is no surfactant, the oil phase is immovable, thus it is natural to assume the displacement to be one-phase. 
Without loss of generality, one could think of it as of the porous media filled with polymer solution, the injected liquid being pure water. The  
viscosity of a polymer solution is a monotonically increasing function of the concentration of polymer; hence the viscosity of injected water is lower than   
viscosity of polymer solution, which would generate viscous fingers. The polymer/water viscosity ratio and the dependence of viscosity on the concentration of dissolved polymer plays a critical role for the displacement pattern. For further improvements in the design of the post-flush injection scheme see~\cite{Claridge, GVB,tikhomirov2021spe}.

Laboratory experiments and numerical simulations, see e.g. \cite{Nijjer2018, Linear} and references therein, show that if the P\'eclet number is high, then the mixing zone grows linearly in time, and its leading and trailing edges move with asymptotically constant speeds $v^f$ and $v^b$, respectively. Several models are proposed in the literature to describe and estimate these speeds (the Koval model \cite{Koval}, the simplified Koval model \cite{Booth}, the Todd--Longstaff model \cite{TL}, etc.). Among the proposed models, we give preference to one --- the TFE model, which offers pessimistic theoretical estimates of speeds under the \emph{transverse flow equilibrium} assumption~\cite{Yortsos}. This model was rigorously analyzed first for gravity-driven fingers in \cite{Otto2005, Otto2006} and later for viscosity-driven fingers in \cite{Yortsos}. It is shown in~\cite{Linear} that the Koval model does not always provide the lower estimate for the speed of the trailing edge, and the TFE model, while giving the necessary estimates for leading and trailing edges, sometimes looks too pessimistic. The latter suggests that the TFE model requires some refinement.

We suggest an approach to such a refinement based on intermediate polymer concentration in the fingers. The TFE speed estimate includes an integral over the concentration range interval. We propose that this interval could be reduced, due to the fact that  in the viscous fingers pattern not all concentrations are represented near the mixing zone's limits. If we take an integral over just the concentrations present near the shock, it could produce more accurate estimates. Detailed analysis of numerical simulation appeared to be non-trivial, including necessity of certain special considerations in each individual case (see Sect.~\ref{sect4}). We used a different strategy to justify our approach. We provided a rigorous theorem based on idea on intermediate concentrations for a simplified TFE model which is believed to well-approximate miscible displacement in porous media and detailed analysis of several numerical simulations which align well to our hypothesis. Another theoretical justification for the presence of intermediate concentration for a simplified two-tubes model of gravity-driven fingers was given in \cite{PTY}.

In the paper we consider viscous fingers in homogeneous media. 
For a study of the effect of heterogeneity of the medium to the mixing zone formed by viscous fingers see \cite{CGS}; the above mentioned  Koval model was also developed for heterogeneous media.

The structure of the paper is as follows. In Sect.~\ref{sect2} we propose an improved version of the TFE estimates, formula~\eqref{TFE-estimate-mod-estimates}, based on the assumption that there exists an intermediate concentration near the tip of the finger. Additionally, we provide a theoretical justification (Theorems~\ref{Theorem_apost} and~\ref{Theorem_apost_back}, see Appendix~A for the proof) for our improved estimates based on the maximum principle by constructing partial one-dimensional upper and lower bounds to the solution under the TFE assumption. Sect.~\ref{sect3} contains the description of the numerical model used in our simulations and the parameters of individual simulation runs. In Sect.~\ref{sect4} we analyze the results of numerical simulations and compare them to our expectations based on theoretical estimates, see Appendix B for analysis of two simulations with higher precision. Sect.~\ref{conclusion} consists of our conclusions and further discussion.

\section{Theoretical basis for TFE estimate improvement}
\label{sect2}
We study fully miscible flow in porous media under the 
macroscopic approximation of the mechanics of a continuous medium (see e.g.~\cite[Chap.~1,5]{bedrikovetsky2013} or~\cite[Sect.~2]{scovazzi2017}). We consider 
a two-dimensional space domain $(x,y)\in[0,L] \times[0,H]$ with length $L$ and height $H$ represented by the isotropic porous medium with porosity~$\phi$ and permeability~$k$. The domain is initially saturated with water with a solvent (polymer) at its maximal concentration $c_{\max}$, and maximal viscosity $\mu(c_{\max})$. At the left boundary, $x=0$, which represents the injection well, pure water is injected at a constant flux rate $Q$; the concentration of polymer 
in water is equal to 0, and its viscosity is $\mu(0)$ (minimal). The diffusivity of the polymer is $D$ and gravity is neglected. We assume the fluids to be incompressible, the polymer to be fully miscible with water, and water with solvent to be a Newtonian fluid.

We employ the two-dimensional model originally proposed for solvent flooding of oil fields \cite{Peaceman1962}, hereinafter referred to as Peaceman model. It couples the solvent transport equation, the incompressibility condition and Darcy's law for incompressible fluid flow. 
A similar model was considered in \cite{chen2001, dewit2005, mishra2008, pramanik2016, sharma2021} in various assumptions and geometries.
\begin{align}
    \label{01-Peaceman_model}
    &\phi\frac{\partial c}{\partial t} + \mathrm{div}({\bm{q}}c)=D\Delta c,
    \\
    &\label{02-Peaceman_model}
    \bm{q}=-\frac{k}{\mu(c)}\nabla p=-m(c)\nabla p,\qquad \mathrm{div}(\bm{q})=0,
\end{align}
where $c$ is the concentration of the polymer, $\bm{q}$ is the flux,  $p$ is the pressure,  $\mu=\mu(c)$ is the viscosity function, which is assumed to be positive and increasing. The mobility function is $m(c)=k/\mu(c)$. We supply the equations~\eqref{01-Peaceman_model}, \eqref{02-Peaceman_model} with the no-flow boundary condition at $y=\{0,H\}$; and a constant injection flux at the injection well ($x=0$), specifically $\bm{q}(t,0,y)=(Q,0)$. 
By scaling the variables, we can set $L=1$, $Q=1$ and $\phi=1$, and we will assume so in this section. 

The initial condition is $c=c_{\max}$, indicating that prior to injection the area of the reservoir under consideration is filled with the maximum concentration of polymer. The injected fluid is pure water ($c = 0$), hence the polymer concentration never exceeds $c_{\max}$ at any point because the injected water only dissolves the polymer; the equations \eqref{02-Peaceman_model} do not change $c$, while the transport equation \eqref{01-Peaceman_model} will not allow $c$ to exceed the current maximum concentration throughout the field.

Under the assumption of transverse flow equilibrium (TFE, see Remark \ref{rem:TFE} below), rigorous estimates for the velocities of the edges of the mixing zone were made in \cite{Otto2006, Yortsos}. According to this model, the velocities of the front and back ends of the mixing zone can be estimated as follows:
    \begin{align}
    \label{TFE-estimate}
        v^f \leqslant\frac{\overline{m}(0,c_{\max})}{m(c_{\max})} , 
        \qquad v^b\geqslant \frac{\overline{m}(0,c_{\max})}{m(0)},
    \end{align}
    where $\overline{m}$ is the mean value of the function $m$ on the interval $(a,b)$:    
    \begin{align*}
        \overline{m}(a,b)=\frac{1}{b-a}\int\limits_a^b m(c)\,dc.
    \end{align*}

\begin{remark}\label{rem:TFE}
   The TFE model is derived under the assumption that pressure $p$ depends  primarily on the $x$ coordinate and is approximately constant with respect to $y$. In that case the system~\eqref{01-Peaceman_model}--\eqref{02-Peaceman_model} could be replaced by the following one (see \cite{Otto2006, Yortsos} for more details): 
   \begin{align}
   \label{03-TFE-1}
   &\frac{\partial c}{\partial t} + \mathrm{div}({\bm{q}}c)=D\Delta c,
   \\
   \label{03-TFE}
   & {\bm{q}} = (q^x, q^y), \qquad q^x = 
    \dfrac{m(c)}{\displaystyle \frac{1}{H}\int_0^H\!\!\! m(c)\,dy}, \qquad \mathrm{div}({\bm{q}})=0.
\end{align}
The system of equations \eqref{03-TFE-1}, \eqref{03-TFE} is used when referring to inequalities \eqref{TFE-estimate} and further estimates \eqref{TFE-estimate-mod-1}. 
All numerical simulations in this paper are performed for the Peaceman model~\eqref{01-Peaceman_model},~\eqref{02-Peaceman_model}.

\end{remark}

As demonstrated by numerical simulations in \cite{Linear} the estimates \eqref{TFE-estimate} give plenty of margin for the fronts' speed. At the same time, according to \cite[Sect.~3]{Otto2006}, the estimates in the proof of \eqref{TFE-estimate} are sharp at the tip of the fastest finger, provided that the concentration at the tip of the finger falls from $c_{\max}$ to 0.
These arguments suggest that at the tip of the finger there should exist some intermediate concentration. The cursory analysis of concentrations in the numerical simulation confirms this conjecture (see Fig.~\ref{fig:onefinger} and also Sect.~\ref{sect3} for how it was obtained). Note that at the tip of the finger one observes a concentration shock with a distinctly non-zero lower bound.

\begin{figure}
    \centering
    \includegraphics[width = 0.45\textwidth]{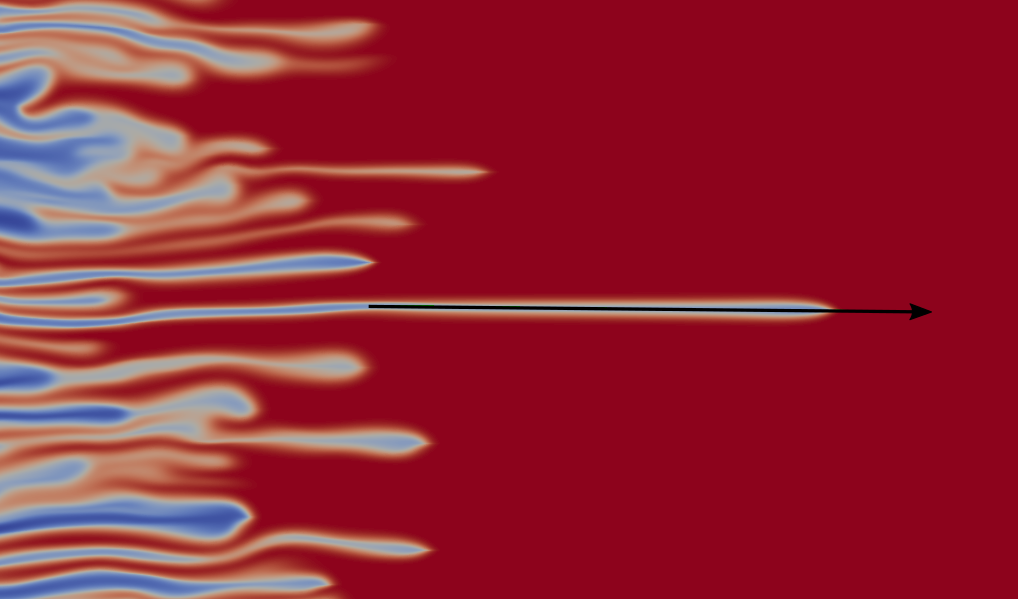}
    \includegraphics[width = 0.45\textwidth]{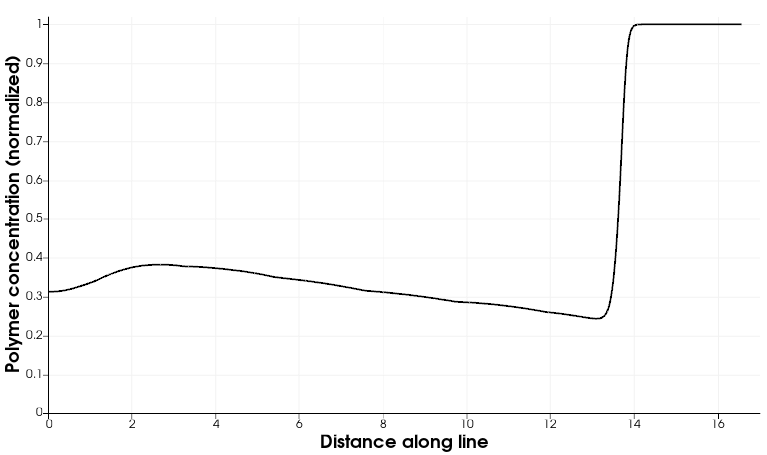}
    \caption{Profile of polymer concentration along a finger. The example shown was obtained
    from a simulation with the linear viscosity model, $M=5$ (see Sect.~\ref{sec:viscosity})}
    \label{fig:onefinger}
\end{figure}

However, given that fingers are constantly restructuring and aren't always perfectly straight, it is challenging to define what the ``tip'' of a finger is. To determine the intermediate concentration, we suggest employing a more indirect method given below.

As mentioned above, TFE estimates are usually too pessimistic, therefore it is natural to assume that the interval of integration in \eqref{TFE-estimate} might be reduced. To verify this assumption, we calculated the front and back velocities of the concentration level lines. We expect (see Fig.~\ref{fig:conjectureillustrated}) that the front velocity would be approximately constant for concentrations in some interval $[c^*,c_{\max}]$ and drop noticeably for lower concentrations. This would empirically show that we have something approximating a shock between concentrations $c^*$ and $c_{\max}$ near the tip of a single finger. Similarly, we expect the back velocities to be approximately constant in some interval $[0, c_*]$. We tested these conjectures against a variety of numerical experiments, and found a degree of agreement with our expectations. Therefore, we propose the use of a modified TFE model to estimate
\begin{align}
\label{TFE-estimate-mod-1}
    v^f \sim \frac{\overline{m}(c^*,c_{\max})}{m(c_{\max})}, \quad v^b \sim \frac{\overline{m}(0, c_*)}{m(0)}.  
\end{align} 
\begin{remark}
    Theorems \ref{Theorem_apost}, \ref{Theorem_apost_back} below prove one-sided estimates
    \begin{align}
    \label{TFE-estimate-mod-estimates}
    v^f \leqslant \frac{\overline{m}(c^*,c_{\max})}{m(c_{\max})}, \quad v^b \geqslant \frac{\overline{m}(0, c_*)}{m(0)}.
\end{align} 
At present, there is only numerical evidence for \eqref{TFE-estimate-mod-1}. Further theoretical understanding of the problem is required to prove the sharpness of the estimates \eqref{TFE-estimate-mod-estimates}. For a less rigorous analysis of a similar question see \cite[Sect.~3]{Otto2006}.
\end{remark}

\begin{figure}[htbp]
    \centering    
    \includegraphics[width = 0.45\textwidth]
    {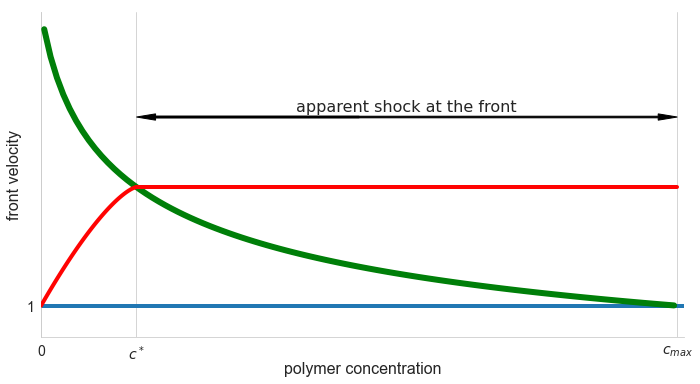}
    \includegraphics[width = 0.45\textwidth]
    {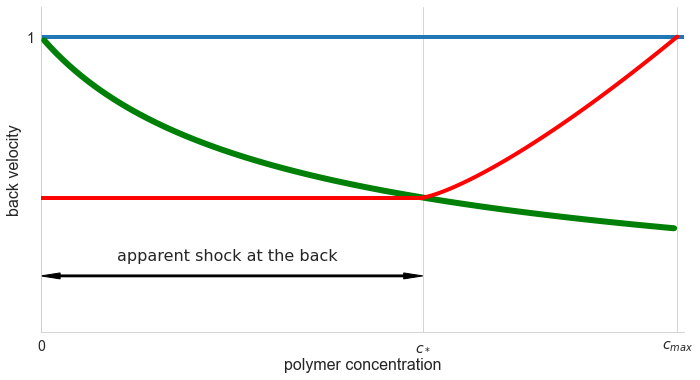}
    \caption{Illustration for our expectations for the level line velocities. Green lines correspond  to the graph of conjectured theoretical approximation of $v^f$ and $v^b$, defined by formula~\eqref{TFE-estimate-mod-1}, as a function of $c^*$ and $c_*$, respectively. Red lines correspond to expected true values for the front speed $v^f$ and back speed $v^b$ of the level lines for the concentration $c\in[0,c_{\max}]$. The conjecture in terms of the left figure states as follows: after the intersection point of the red and green curves, the red one is constant. Similarly, for the right figure: before the intersection point of the red and green curves, the red one is constant. The blue line corresponds to velocity equal to 1.}
    \label{fig:conjectureillustrated}
\end{figure}

To better justify these estimates, we constructed an argument similar to that in \cite{Otto2006, Yortsos}, 
where a lower bound $c_{lb}$  to the solution of the TFE system \eqref{03-TFE-1},\eqref{03-TFE} was proposed, satisfying the auxiliary equation
\begin{equation}
\label{Lower_bound_equation}
\dfrac{\partial c_{lb}}{\partial t} + \dfrac{m(c_{lb})}{m(c_{\max})}\dfrac{\partial c_{lb}}{\partial x} = D \dfrac{\partial^2 c_{lb}}{\partial x^2}.
\end{equation}
With the help of the maximum principle it was proved that if $c-c_{lb} > 0$ at $t=0$, then it remains positive at all times. We used a similar idea to prove the following a posteriori estimate. 

\begin{theorem}
\label{Theorem_apost}
Let $c(t, x, y)$ be a solution of the TFE model \eqref{03-TFE-1}, \eqref{03-TFE} in the strip $(x,y)\in\mathbb{R}\times[0,H]$, with no-flow boundary conditions at $y=\{0,H\}$, that is $q^y(t,x,0)=q^y(t,x,H)=0$, and initial data satisfying
\begin{align}
\label{initial-data-theorem}
\lim\limits_{x\to-\infty}\max\limits_y  c(0,x,y)=0, \qquad \lim\limits_{x\to+\infty}\min\limits_y  c(0,x,y)=c_{\max}.
\end{align}
Denote by 
\begin{align}
\label{x-front}
    x^f(t,r) := \sup\{ x\in\mathbb{R}: c(t,x,y)=r \text{ for some } y\in[0,H]  \}
\end{align}
the front tip of the level line for the concentration $r$.
Suppose there is a concentration $c^*$, such that
\begin{equation}
\label{a_post_est}
x^f(t, c^*) \leqslant \dfrac{\overline{m}(c^*, c_{\max})}{m(c_{\max})} \cdot t + l_1
\end{equation}
for some constant $l_1$ and all $t>0$. Then for any $\widetilde{c} \in (c^*, c_{\max})$ and any $v>\dfrac{\overline{m}(c^*, c_{\max})}{m(c_{\max})}$ there exists a constant $l_2> 0$ that depends on $v$,  such that as $t\to+\infty$  we have
\begin{equation}
\label{a_post_est1}
x^f(t, \widetilde{c}) \leqslant vt + l_2.
\end{equation}
\end{theorem}
\begin{remark}
\label{remark-estimate}
With appropriate definition of the front velocity for concentration level $c$, estimate \eqref{a_post_est1} 
could be interpreted as
\begin{equation}\label{eq-rem2}
v^f(c) \leqslant \dfrac{\overline{m}(c^*, c_{\max})}{m(c_{\max})}, \quad c \geqslant c^*.   
\end{equation}
\end{remark}

The estimate for the back velocity $v^b$ is justified with a similar construction of an upper bound auxiliary equation solution and a similar theorem.

\begin{theorem}
\label{Theorem_apost_back}
Let $c(t, x, y)$ be a solution of the TFE model \eqref{03-TFE-1}, \eqref{03-TFE} in the strip $(x,y)\in\mathbb{R}\times[0,H]$, with no-flow boundary condition at $y=\{0,H\}$ and initial data satisfying~\eqref{initial-data-theorem}.
Denote by 
$$
x^b(t,r) := \inf\{ x\in\mathbb{R}: c(t,x,y)=r \text{ for some } y\in[0,H]  \}
$$ 
the back tip of the level line for the concentration $r$.
Suppose there is a concentration $c_*$, such that
\begin{equation}\notag
x^b(t, c_*) \geqslant \dfrac{\overline{m}(0, c_*)}{m(0)} \cdot t - l_1
\end{equation}
for some constant $l_1$ and all $t>0$. Then for any $\widetilde{c} \in (0, c_*)$ and any $v<\dfrac{\overline{m}(0,c_*)}{m(0)}$ there exists a constant $l_2> 0$ that depends on $v$,  such that as $t\to+\infty$  we have
\begin{equation}
\label{a_post_est_back}
x^b(t, \widetilde{c}) \geqslant v t - l_2.
\end{equation}
\end{theorem}
\begin{remark}
With appropriate definition of the back velocity for concentration level $c$, estimate \eqref{a_post_est_back} 
could be interpreted as
\begin{equation}\label{eq-rem3}
v^b(c) \geqslant \dfrac{\overline{m}(0, c_*)}{m(0)}, \quad c \leqslant c_*.    
\end{equation}
\end{remark}

In what follows (see Sect.~\ref{sect4}) we demonstrate that the approximations \eqref{TFE-estimate-mod-1} are in good agreement with the numerical experiments in most cases. We provide proof of Theorem~\ref{Theorem_apost} 
in Appendix~A. Theorem~\ref{Theorem_apost_back} is proved similarly.

\section{Description of the numerical model}
\label{sect3}

\subsection{General considerations}
A qualitative model has been set up to check the hypothesis of intermediate concentration and
corroborate the theoretical findings developed in Sect.~\ref{sect2}.
Some parameters, which are believed not to affect the qualitative results, were chosen to have values corresponding to a realistic oil field (see Table~\ref{tab:param}):
\begin{itemize}
\item the concentration $c_{\max}$ (is a free scaling factor in the equations),
\item the porosity $\phi$ (acts as an inverse distance scaling factor in the equations),
\item the base value of permeability (changes only the pressure field),
\item the base value of viscosity (changes only the pressure field),
\item the injection flux rate (acts as time scaling factor).
\end{itemize}

\begin{table}[h]
 \centering
 \begin{tabular}{l|l}
Reservoir parameters  & Data values\\
 \hline
Initial concentration of polymer, $c_{\max}$ & 0.0015\\
Viscosity of water, $\mu_0$  & 0.001 [cP]\\
Porosity, $\phi$  & 0.188\\
Diffusion coefficient, $D$  & 0 [m$^{2}$/s]\\
Viscosity ratio parameter, $M$ &  5, 10, 20, 40\\
Injection flux, $Q$ &  12 [m$^2$/day]\\
Computational domain length, $L$ & 64 [m]\\
Computational domain height, $H$ & 36 [m]\\
Linear viscosity function, $\mu_\mathrm{linear}(c)$ & $\mu_0\left(1 + (M-1)\frac{c}{c_{\max}}\right)$ [cP]\\
Quadratic viscosity function, $\mu_\mathrm{quadratic}(c)$ & $\mu_0\left(1 + (M-1)\frac{c^2}{c^2_{\max}}\right)$ [cP]\\
Exponential viscosity function, $\mu_\mathrm{exponential}(c)$ & $\mu_0 M^{c / c_{\max}}$ [cP]\\
 \end{tabular}
 \caption{The parameters used to model the miscible displacement.}
 \label{tab:param}
\end{table}

Diffusion causes fingers to dissolve or merge, which may slow down growth of the mixing zone, see recent paper \cite{Laz} for some numerical experiments.
To test the validity of estimates (\ref{TFE-estimate-mod-estimates}), we are interested in keeping the growth of the mixing zone as high as possible. Thus, the diffusion $D$ was set to zero in the model. However, the numerical scheme brings numerical diffusion (see Sect.~\ref{sec:grid}).

\subsection{Viscosity functions}\label{sec:viscosity}
While the typical choice of the viscosity function consists in assuming an Arrhenius-like exponential dependence (see e.g. \cite{Nijjer2018} and references therein), we additionally consider two other viscosity functions and our motivation is as follows. The existing theoretical estimates~\eqref{TFE-estimate} for the length of the mixing zone for the simplified TFE model depend not only on the mobility ratio, but also on the whole mobility curve (or viscosity curve). Similarly, in our idea of refining these estimates the answer~\eqref{TFE-estimate-mod-1} also depends on the whole viscosity curve. Thus, it is crucial to test numerically the conjecture for different viscosity functions and we choose linear and quadratic dependence as a natural model for positive increasing function.

For each of the three types of viscosity functions (exponential, linear, quadratic) the coefficients were chosen so that $\mu(0)=\mu_0$, $\mu(c_{\max})=\mu_0 M$:

$$\mu_\mathrm{linear}(c) = \mu_0\left(1 + (M-1)\frac{c}{c_{\max}}\right),$$
$$\mu_\mathrm{quadratic}(c) = \mu_0\left(1 + (M-1)\frac{c^2}{c^2_{\max}}\right),$$
$$\mu_\mathrm{exponential}(c) = \mu_0 M^{c / c_{\max}}.$$

We consider the values of $M = 5, 10, 20, 40$ in order to test our hypothesis against different real life-like cases. The viscosity of water $\mu_0$ was taken equal to 0.001 cP 
(the value corresponds to temperature of~$\approx 20^\circ$C). 

\subsection{Grid}\label{sec:grid}
The simulated fields is a rectangle of length 64 meters and height 36 meters. (The physical size in fact does not matter because the analysis of mixing zone growth is non-dimensionalized.)
The proportions were chosen arbitrarily, but there were some limitations. We can not have a very ``tall'' reservoir because, while we will have a lot of viscous fingers, we will not have the ability to analyze their growth over time, because they will reach the right edge of the reservoir too soon. Also, we can not have a very ``wide'' reservoir because in such a case the number of fingers will be small and they will dissolve/merge due to diffusion before they reach the right edge.

The grid was rectangular, with step size $\Delta$ equal to 2.5 cm in most 
of our simulations, so the grid had $2560\times 1440 \approx 3.7$  million cells.
The coordinate system for later calculations is normalized so that the length $L = 1$, and the height $H=9/16$. The P\'eclet number is double the horizontal grid size, 5120 in this case.
For other simulations with P\'eclet number in a similar range, see \cite{Nijjer2018}. 

The numerical scheme brings the numerical diffusion equal to $D_{\mathrm{num}} = {v \Delta}/{2\phi}$, or just ${\Delta}/{2}$ in normalized coordinates, which is equal to the inverse of the P\'eclet number. Due to limited computational resources, it was hard to obtain lower values of 
$D_{\mathrm{num}}$. However, two separate simulations 
were performed on a $5120\times 2880$ grid ($\approx 14.7$  million cells), see Appendix~B.

To stimulate the genesis of the viscous fingers, permeability was made varying across grid cells.
A random log-normal distribution was used: $\log(K) = {\cal N}(\log(K_0), 0.01)$, where $K_0 = 80 \text{ mD}$ was chosen.
Porosity was constant and equal to 0.188.

\subsection{Initial and boundary conditions}
The left side of the rectangle acted as an injection well, while the right side acted as a production well. The constant flux boundary condition (12 m$^2$/day) was set on the left side; on the right side, a constant pressure condition was maintained. (The boundary value of the pressure can be arbitrary
and is not relevant, since the equations depend only on the pressure gradient.) The time scale for later calculations was normalized so that the flux $Q = 1$. The no-flow condition was imposed on the upper and lower sides. 

Initially, the reservoir is filled with polymer of concentration ${c_{\max}} = 0.0015$. The injected fluid, as said above, has $c = 0$.

\subsection{Software}

DuMu$^\mathrm{x}$~\cite{flemisch2011dumux,Kochetal2020Dumux} is a porous medium flow framework based on the DUNE (Distributed and Unified Numerics Environment) toolbox~\cite{ans-DUNE}. Both DuMu$^\mathrm{x}$ and DUNE are open-source software and are written in C++ with extensive use of template metaprogramming.

While DUNE provides generic interfaces for the implementation of different discretization schemes,
DuMu$^\mathrm{x}$ uses the finite volume method solely. Both IMPES and fully implicit schemes were developed
in DuMu$^\mathrm{x}$ version 2, but since version 3, the focus has shifted to fully implicit schemes,
while IMPES-based implementations of porous flow are no longer developed. Following that decision, a fully
implicit scheme was used for DuMu$^\mathrm{x}$ simulation in this work. As for the flux approximation,
DuMu$^\mathrm{x}$ provides a cell-centered scheme (both two-point and multi-point flux approximation
flavors) and vertex-centered scheme, also known as the Box method. The latter was used in this work. The
basis functions used in the approximation are linear.
Finally, of the many linear solvers that DUNE provides, the implementation of the biconjugate gradient
stabilized method with the algebraic multigrid preconditioner (AMGBiCGSTAB) was used.

DuMu$^\mathrm{x}$ accepts grids in different formats, including the well-known ``msh'' format.
Gmsh open-source software~\cite{gmsh} was used to produce the rectangular grid in this work.

\subsection{Simulations and estimation of velocities of fronts}\label{sec:simulations}
We ran 12  simulations, four per each of the three viscosity models, with $M=5,10,20,40$ on the $2560 \times 1440$ grid; and two simulations
with $M=5$, one with slow-growing (linear), and the other with fast-growing (exponential) viscosity functions, on the $5120 \times 2880$ grid. 

The time step was 5 minutes of simulation time; steps usually took one to three iterations of the Newton method to converge. Each simulation went on until any of the viscous fingers reached the production well.

Every 2 hours of simulation time, the map of the polymer concentration was analyzed. For each concentration $C$ from an equidistant grid of concentrations

$$C\in\{0.005{c_{\max}},0.015{c_{\max}},\ldots 0.985{c_{\max}}, 0.995{c_{\max}}\}$$

\noindent the positions of the forefront and rear front for that concentration at time $t_i$ are defined as

\begin{equation}\label{eq:forefront-reafront}
\begin{aligned}
\mathrm{forefront}(t_i, C) = \max x : \exists y \in [0.05H, 0.95H]: c(t_i, x, y) \leqslant C\\
\mathrm{rearfront}(t_i, C) = \min x : \exists y \in [0.05H, 0.95H]: c(t_i, x, y) \geqslant C
\end{aligned}
\end{equation}

\noindent where $H = 36$ meters is the height of the simulation field. The decision to cut
off top and bottom 5\% from analysis came from the fact that the no-flow boundary conditions 
on the horizontal borders of the field favor the forming of thin viscous fingers that move faster than ordinary fingers (see Fig.~\ref{fig:borderfinger}).

\begin{figure}[h]
  \begin{center}
    \includegraphics[width=0.45\textwidth]{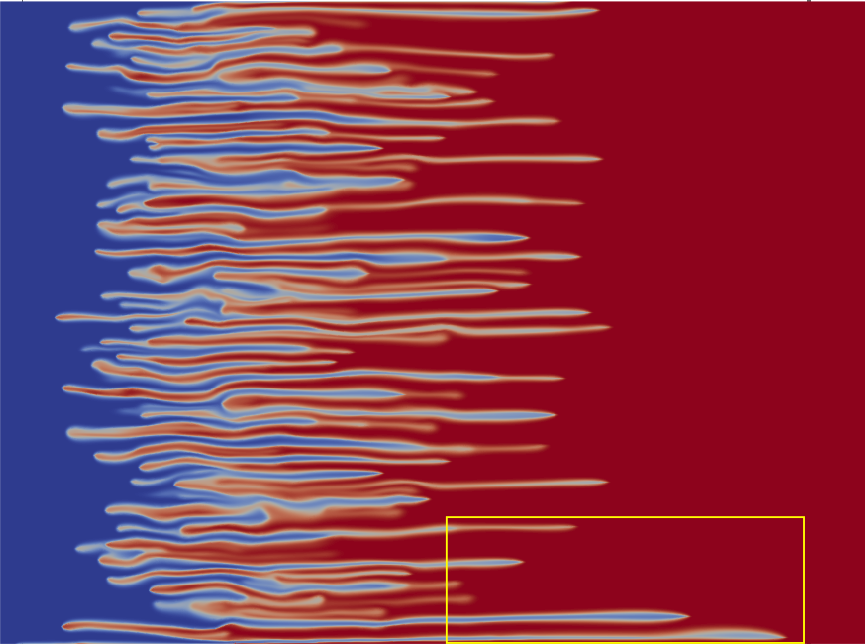}
    \includegraphics[width=0.45\textwidth]{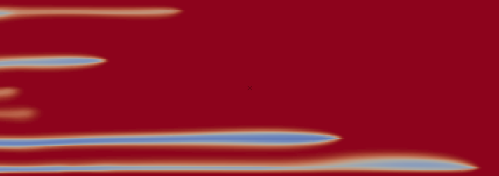}
    \end{center}
  \caption{An example of a viscous finger that is close
  to the bottom border and moves faster than other fingers. On the right picture an enlarged region is shown. The picture was obtained during a simulation with the quadratic viscosity model, $M = 10$. Several other simulations have thin fingers on horizontal borders.}\label{fig:borderfinger}
    \end{figure}
  
The average velocities of fronts were estimated from the time sequence of the positions of fronts via linear regression without intercept:

\begin{equation}\label{eq:fronts-regression}
\begin{aligned}
v^f(C) \approx \frac{\sum_i t_i\,\mathrm{forefront}(t_i, C)}{\sum_i t_i^2},& \quad
v^b(C) \approx \frac{\sum_i t_i\,\mathrm{rearfront}(t_i, C)}{\sum_i t_i^2}
\end{aligned}
\end{equation}
where, as noted above, $t_i = t_{i-1} + 2\ \mathrm{hours}$.

\section{Analysis of the results}
\label{sect4}

\subsection{Forward front analysis}\label{sec:forward-front-12}

Since the trajectory for a given concentration of the front is not exactly a straight line (see Fig. \ref{fig:front}), the question arises of how to determine the speed of the front. The most natural approach is to define the value $v^f(C)$ as a linear regression, see \eqref{eq:fronts-regression}. At the same time, following the spirit of Theorem \ref{Theorem_apost} it is helpful to consider 
\begin{equation}\label{eq:fronts-max}
v^f_{\max}(C) = \max_{t > 0.05} \frac{\mathrm{forefront}(t, C)}{t}.
\end{equation}

\begin{figure}[h!]
    \centering
    \includegraphics[width=0.6\textwidth]{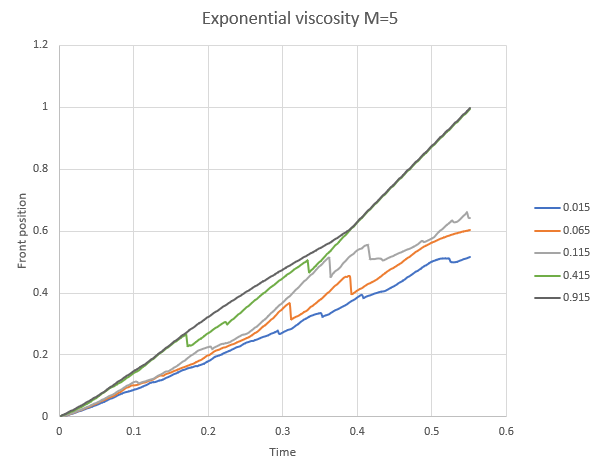}
    \caption{Position of the front for $C/c_{\max} = 0.015, 0.065, 0.115, 0.415, 0.915$, for exponential viscosity with $M = 5$}
    \label{fig:front}
\end{figure}

Results before $t = 0.05$ (in normalized units) are excluded because of the dominance of numerical diffusion over viscous fingering at early steps of simulation.

In Fig. \ref{fig:front-analysis} we plot the front velocities as functions of concentration $C$ calculated by \eqref{eq:fronts-regression}, \eqref{eq:fronts-max} and the speed determined by the TFE estimate ${\overline{m}(C, c_{\max})}/{m(c_{\max})}$, see \eqref{TFE-estimate-mod-1}. Note that since the first two curves are monotonically non-decreasing, each of them should have only one point of intersection with the third curve, which is monotonically decreasing (being a mean value of an integral of a decreasing function $m(c)$). 
Let $c^*_0$ be the $c$-coordinate of the intersection point for either $v^f_{\max}$ or $v^f$ and the TFE estimate. Then Theorem \ref{Theorem_apost} suggests that the velocity function should be approximately constant for $C>c^*_0$.

\begin{figure}[h]
\centering
\includegraphics[width = 0.32\linewidth]{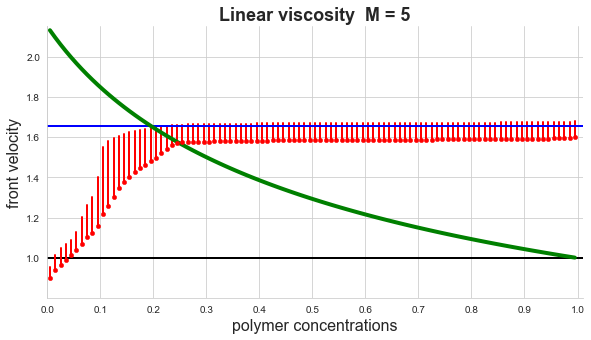} 
\includegraphics[width = 0.32\linewidth]{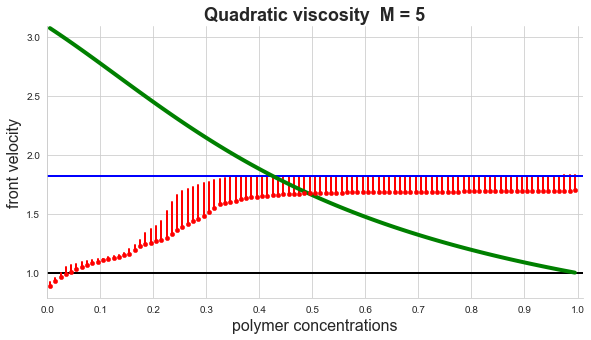} 
\includegraphics[width = 0.32\linewidth]{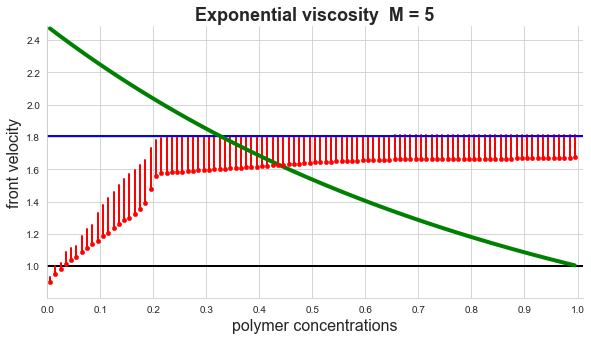}

\includegraphics[width = 0.32\linewidth]{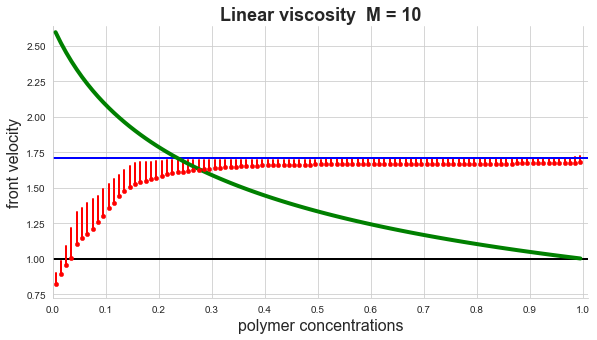} 
\includegraphics[width = 0.32\linewidth]{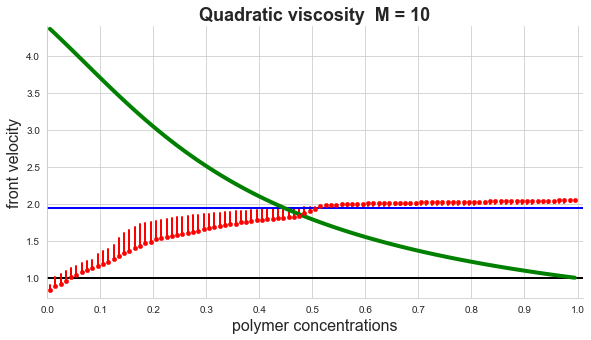}
\includegraphics[width = 0.32\linewidth]{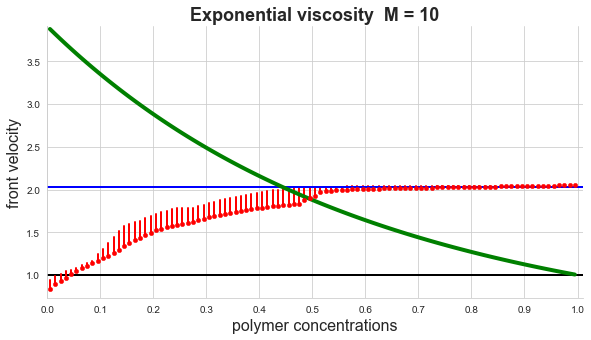}

\includegraphics[width = 0.32\linewidth]{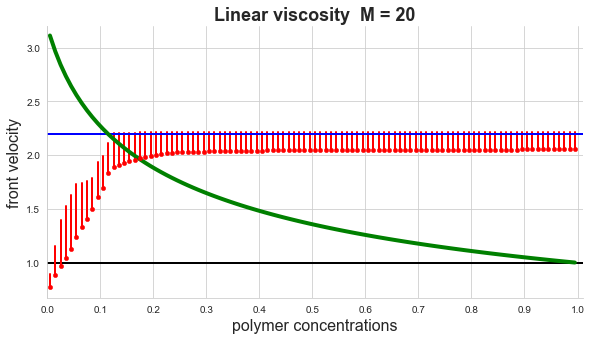} 
\includegraphics[width = 0.32\linewidth]{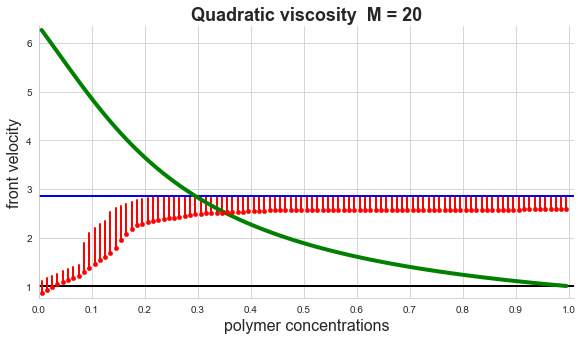} \includegraphics[width = 0.32\linewidth]{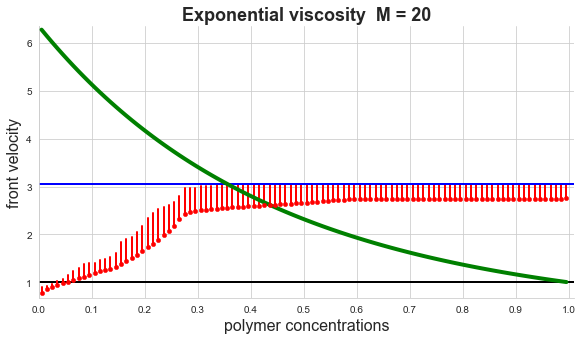}

\includegraphics[width = 0.32\linewidth]{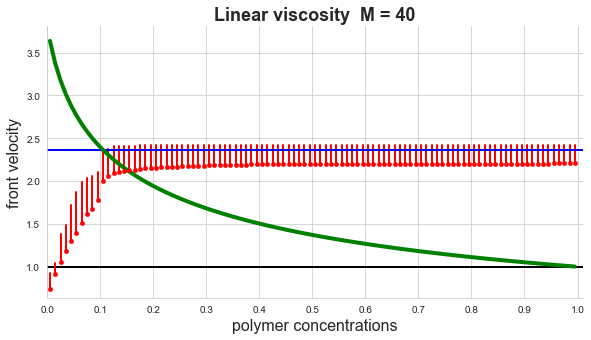} 
\includegraphics[width = 0.32\linewidth]{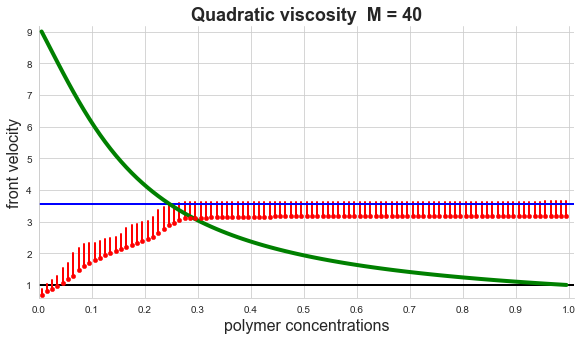} 
\includegraphics[width = 0.32\linewidth]{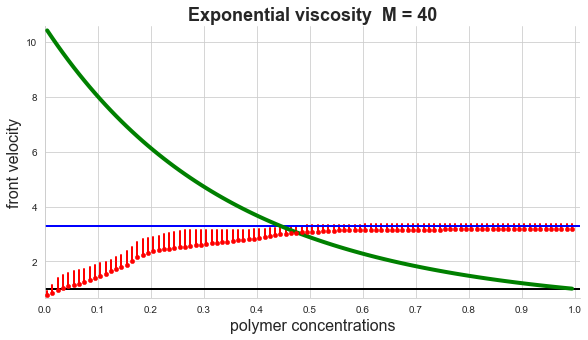}

\caption{Comparison of the speed of the fastest finger with estimates by \eqref{TFE-estimate-mod-1}. The green curve corresponds to \eqref{TFE-estimate-mod-1}, filled dots correspond to \eqref{eq:fronts-regression}, upper point of the vertical strip corresponds to \eqref{eq:fronts-max}. The vertical strip shows the difference between \eqref{eq:fronts-regression} and \eqref{eq:fronts-max}. The results are shown for linear (left column), quadratic (middle column), and exponential (right column) viscosity with contrasts (from top to bottom) $M = 5, 10, 20, 40$.} \label{fig:front-analysis}

\end{figure}

\begin{remark}
    We would like to stress two moments that make Theorems \ref{Theorem_apost} and \ref{Theorem_apost_back} not strictly applicable to the plots on Fig.~\ref{fig:front-analysis}. First, Theorem \ref{Theorem_apost} is valid for the TFE approximation, while the simulations are made with the Peaceman model. There is no known analytical relation between the TFE and Peaceman models. Second, estimate \eqref{a_post_est} has the constant term $l_1$, while expressions in \eqref{eq:fronts-regression} assume linear regression and \eqref{eq:fronts-max} does not take $l_1$ into account.
\end{remark}

The aim of the paper is to understand the role of intermediate concentrations in the speed of viscous fingers. We distinguish between three different situations:
\begin{itemize}
    \item[\textbf{A:}] The velocity function is approximately constant for $C>c^*_0$ and has a significant drop for $c<c^*_0$ and fulfilled the estimate \eqref{eq-rem2}.
    \item[\textbf{B:}] The velocity function is approximately constant for $C>c_1$, where $c_1<c^*_0$  and fulfilled the estimate~\eqref{eq-rem2}.
    \item[\textbf{C:}] The velocity function is not approximately constant for $C>c^*_0$ or not satisfy the inequality \eqref{eq-rem2}.    
\end{itemize}
The cases \textbf{A, B, C} are not defined rigorously. Sometimes we attribute a simulation to two cases simultaneously, when the exact decision cannot be made. In Tables \ref{tab:summ-avg} and \ref{tab:summ-max} we summarize the results of the simulations for velocities $v^f(C)$ and $v^f_{\max}(C)$, respectively.

\begin{table}
\centering
\parbox{.45\linewidth}
{
    \centering
    \begin{tabular}{c|c|c|c}
       & linear & quadratic & exponential  \\ \hline
M = 5  & A & A, B & B  \\
M = 10 & A, B & C &  A, C \\
M = 20 & A, B & B &  B\\
M = 40 & A & A &  A
\end{tabular}
    \caption{Role of intermediate concentration for average velocity}
    \label{tab:summ-avg}
}
\hspace{10pt}
\parbox{.45\linewidth}
{
    \centering
    \begin{tabular}{c|c|c|c}
       & linear & quadratic & exponential  \\ \hline
M = 5  & A & A, B & B  \\
M = 10 & A & C & A \\
M = 20 & A & B & A, B\\
M = 40 & A & A & A
\end{tabular}
    \caption{Role of intermediate concentration for maximum velocity}
    \label{tab:summ-max}
}
\end{table}

Note that for both velocity curves $v^f(C)$ and $v^f_{\max}(C)$ the right part of the velocity graph is nearly flat.
The difference between $v^f$ and $v^f_{\max}$ could be considered as uncertainty of velocity, as the position of the front is not exactly linear.

In most of the cases the graph of the average velocity $v^f(C)$ changes smoother to the left of the flat part  than for the maximal velocity $v^f_{\max}(C)$. This makes it less valuable as it is more difficult to identify the left boundary of the flat part, therefore increasing the number of mixed labels \textbf{A, B} and \textbf{A, C} in Table \ref{tab:summ-avg} compared to Table \ref{tab:summ-max}. 
In some of the cases velocity has a mild incline on the right side of the graph and at some intermediate concentration the speed changes much more significantly (for example, quadratic viscosity with $M=20$ and exponential viscosity with $M=5$ and $M=20$). In that case we interpret the part of the graph with a mild incline as the flat part. 
It is possible that the intersection point is significantly lower than the flat part (2 simulations): 
\begin{itemize}
\item for exponential viscosity with $M=10$ the distance in $C$-coordinate between the intersection point and $c_0^*$ is approximately 0.02, which we considered insignificant and marked \textbf{A, C};  
\item for quadratic viscosity with $M=10$ the distance in $C$ is around 0.04, which we considered more significant and marked this case as \textbf{C}. 
\end{itemize}
If the intersection point is inside the flat part (including mildly inclined cases) we marked it \textbf{B} or \textbf{A, B} depending on how close the intersection point is to the boundary of the flat region. Other cases we marked~\textbf{A}. 

For the maximal velocity we also observe a long flat part in the right end of the graph. It is higher and typically longer than the flat part of $v^f$, so the intersection point with the curve \eqref{TFE-estimate-mod-1} is located to the left of the corresponding intersection with the average speed curve. The decrease from the flat part of the graph is typically more visible. Note that the intersection point with \eqref{TFE-estimate-mod-1} in the simulation for quadratic viscosity with $M=10$ is to the left from and significantly lower than the flat part of $v^f_{\max}$ (marked \textbf{C}). In most of the remaining cases it either coincides with (7 simulations, marked \textbf{A}) or appears to the right from (2 simulations, marked \textbf{B}) the smallest value of $C$ on the flat part. The cases of quadratic viscosity with $M=5$ and exponential with $M=20$ we marked as \textbf{A, B} since the speed has very mild decrease for $C<c_0^*$ and a much more significant decrease for $C< c_1$ for some $c_1<c_0^*$.

We give the following interpretation to the cases 
\begin{itemize}
    \item[\textbf{A:}] The intermediate concentration in fingers is the main mechanism for further speed drop compared to TFE estimates.
    \item[\textbf{B:}] There exists a significant drop in the speed comparing to TFE estimates, but intermediate concentration is not the only mechanism of speed reduction.
    \item[\textbf{C:}] Either the numerical simulation is of low quality or the TFE model is not a good approximation of the original equations.    
\end{itemize}

Recall that case \textbf{C} does not automatically mean low simulation quality, since equations \eqref{a_post_est}, \eqref{a_post_est1} used different constants $l_1$ and $l_2$ and the distance between the flat part and the point of the intersection is relatively small. At the same time, the decrease of the speed near the intersection point is non-typical and demands further investigation.

Below we consider in more details the most common representatives of the cases
\begin{itemize}
    \item[\textbf{A:}] linear viscosity with $M=20$.
    \item[\textbf{B:}] exponential viscosity with $M=5$.
    \item[\textbf{C:}] quadratic viscosity with $M=10$.    
\end{itemize}

For linear viscosity with $M=20$ (label \textbf{A}) we plot the level lines of concentration at a fixed time moment for 3 concentrations $c_1 < c_2\sim c^* < c_3$, Fig. \ref{fig:Lin20-3levels}.
We see that on 2 fastest fingers the level lines corresponding to $c_2$ and $c_3$ almost coincide, and the line corresponding to $c_1$ is located significantly to the left of them. Such a behavior fits our conjecture. In the third fastest finger we see a small island of concentration $c_1$, which seems to disappear soon, which could illustrate the mechanism of forming the intermediate concentration. In particular, we also observe several islands disconnected from other fingers, but they do not have an effect on the speed of the fastest finger. 

For exponential viscosity with $M=5$ (label \textbf{B}) we plot the level lines of concentration at a fixed time moment for 3 concentrations $c_1 < c_2 < c_3 \sim c^*$, Fig. \ref{fig:Exp5-3levels}.
We see that on the fastest finger level lines corresponding to $c_2$ and $c_3$ almost coincide. This is an indicator that the intermediate concentration in the finger does not coincide with $c^*$ and is lower. Note that the line corresponding to $c_1$ is located significantly to the left of the lines for $c_2$ and $c_3$, which tells us that the intermediate concentration exists, but does not coincide with $c^*$. 

\medskip
\begin{minipage}{0.49\textwidth}
    \centering
    \includegraphics[width = 0.95\textwidth]{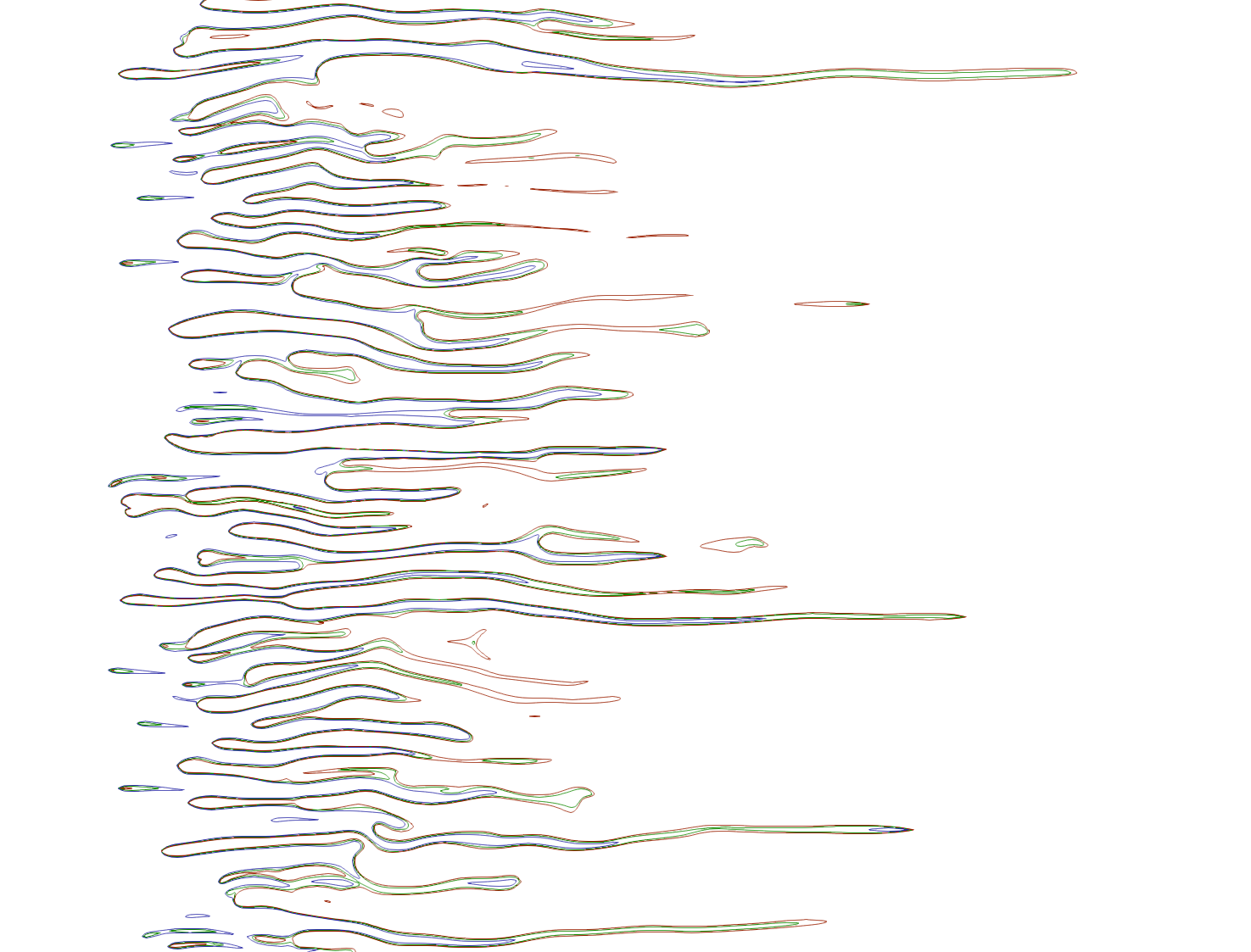}
    \captionof{figure}{Level lines for concentrations $c_1/c_{\max} = 0.075 $ (blue), $c_2/c_{\max}\sim c^*/c_{\max} = 0.125$ (green), $c_3/c_{\max} = 0.175$ (red).}
    \label{fig:Lin20-3levels}
\end{minipage}\hspace{20pt}
\begin{minipage}{0.49\textwidth}
    \centering
    \includegraphics[width = 0.95\textwidth]{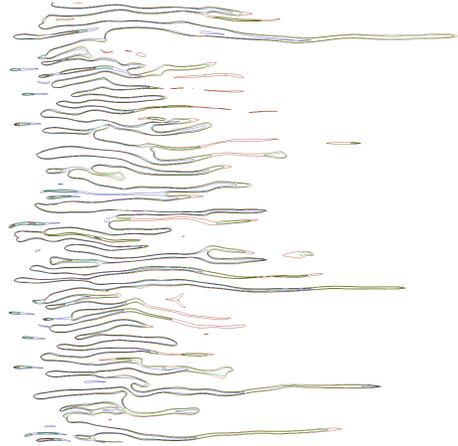}
    \captionof{figure}{Level lines for concentrations $c_1/c_{\max} = 0.125$ (blue), $c_2/c_{\max} = 0.225$ (green), $c_3/c_{\max} \sim c^*/c_{\max} = 0.325$ (red).}
    \label{fig:Exp5-3levels}
\end{minipage}
\medskip

The most important is the analysis of the case corresponding to label \textbf{C}, since it does not agree well with Theorem~\ref{Theorem_apost}. In Fig.~\ref{fig:quad10-2dim} we show the two-dimensional profile of the concentration at a fixed moment of time.
\begin{figure}[h]
    \centering
    \includegraphics[width=0.6\textwidth]{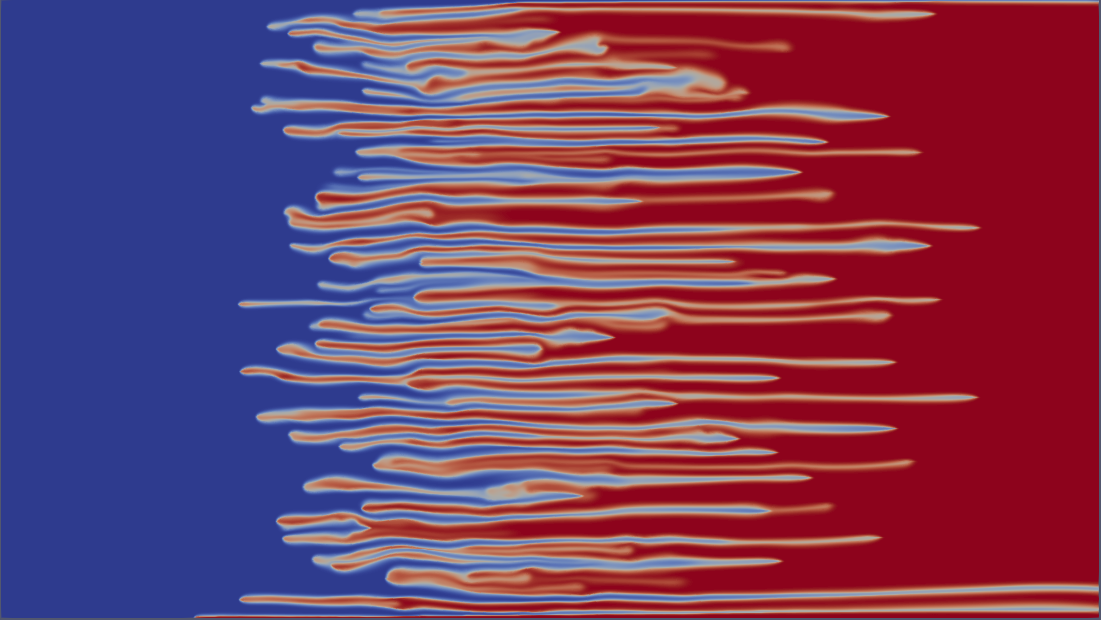}
    \caption{Value of concentrations at a fixed time moment for quadratic viscosity with $M=10$. Blue corresponds to $c=0$, red corresponds to $c=c_{\max}$.}
    \label{fig:quad10-2dim}
\end{figure}
In this picture one can see two  fast fingers near the bottom boundary. Since the behavior near the boundary can be significantly different from the behavior in the main part of the strip, we cut out the top and bottom 5\% of the strip \eqref{eq:forefront-reafront}. However, one of those two fingers was not cut out. This potentially affects the speed estimates. For this simulation we perform an analysis similar to Figure \ref{fig:front-analysis} cutting 7\% from the boundaries:
$$
\mathrm{forefront}(t_i, C) = \max x : \exists y \in [0.07H, 0.93H]: c(t_i, x, y) \leqslant C.
$$
This has allowed us to exclude the both fast fingers near the bottom boundary. The resulting graph of average and maximal velocities is shown in Figure \ref{fig:quad10-mod}.
\begin{figure}
    \centering
    \includegraphics[width = 0.6\textwidth]{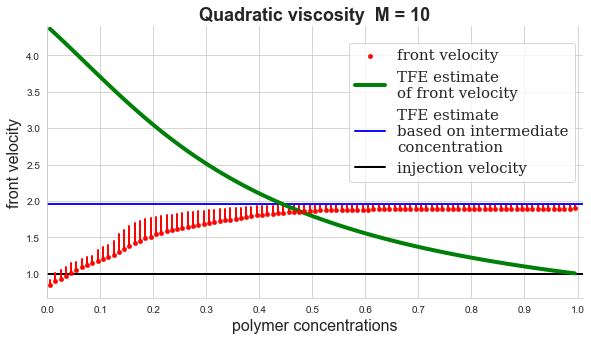}
    \caption{Comparison of the speed of the fastest finger with estimates by \eqref{TFE-estimate-mod-1} for quadratic viscosity with contrast $M=10$ with 7\% cut from the boundary.}
    \label{fig:quad10-mod}
\end{figure}
We see that the left point of flat part of velocity coincides with the point of intersection with estimates \eqref{TFE-estimate-mod-1} for both maximal and average velocities. We conclude that after the 7\% cutoff, the simulation for the quadratic viscosity with $M=10$ should be labeled~\textbf{A}.

\subsection{Rear front analysis}\label{sec:back-front-12}

As in the case of the forefront, we start with investigation of the position of the back front depending on concentration, see Fig. \ref{fig:back}.
\begin{figure}[h!]
    \centering
    \includegraphics[width=0.3\textwidth]{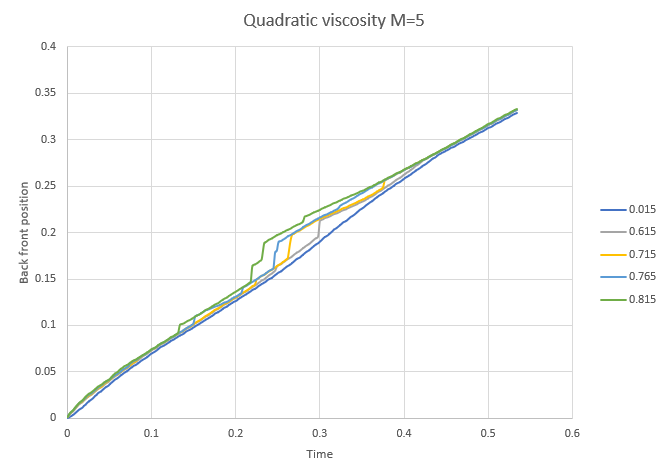}
    \includegraphics[width=0.3\textwidth]{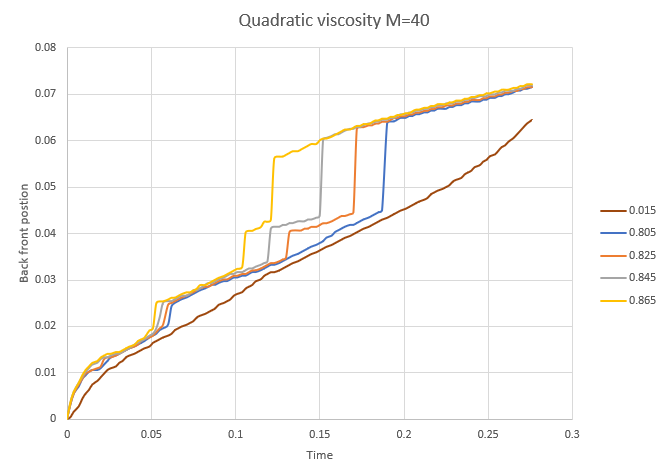}
    \includegraphics[width=0.3\textwidth]{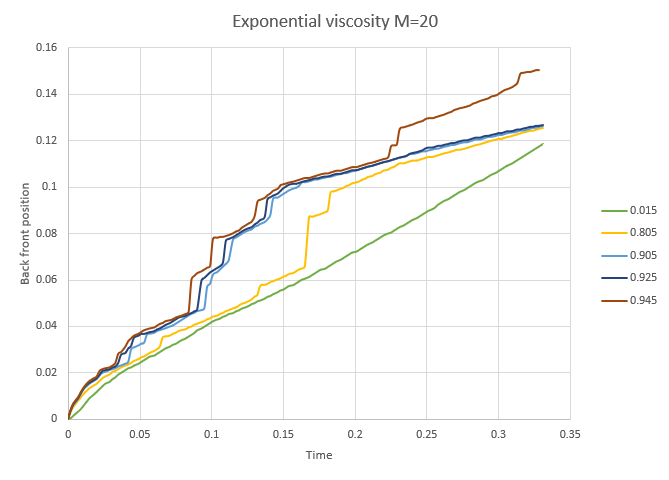}
    \caption{Position of the back front for various concentration levels for $\mu_\mathrm{quadratic}$ with $M=5$, $\mu_\mathrm{quadratic}$ with $M=40$, $\mu_\mathrm{exponential}$ with $M=20$}
    \label{fig:back}
\end{figure}
Note that for low values of concentration the position of the front is almost linear in time (for more details see \cite{Linear}). For higher concentrations, though, the position of the back front has significant jumps and cannot be well approximated by a line. We suggest that those jumps are the consequence of the numerical diffusion, and its effect is larger than in the case of the forward front since back front moves slower. We do not perform a detailed analysis of those jumps. 

Similarly to the forward front and following the spirit of Theorem \ref{Theorem_apost_back} we consider the minimal velocity defined by 
\begin{equation}\label{eq:fronts-min}
v^b_{\min}(C) = \min_t \frac{\mathrm{rearfront}(t, C)}{t}.
\end{equation}

We plot the back front velocities as functions of concentration $C$ calculated by \eqref{eq:fronts-regression} and \eqref{eq:fronts-min} and speed determined by the TFE estimate $\overline{m}(0, C)/m(0)$, see \eqref{TFE-estimate-mod-1}. Since the first two curves are non-decreasing and the third is monotonically decreasing, there should be only one point of intersection for each speed function.  Let $c_{*,0}$ be the $c$-coordinate of the intersection point for either $v^b_{\min}(C)$ of $v^b$ and TFE estimate. According to Theorem \ref{Theorem_apost_back} we expect the velocity function to be approximately constant for $C < c_{*,0}$. Compared to the case of the fastest finger, the back front demonstrates a much lower degree of agreement with estimates \eqref{TFE-estimate-mod-1} inspired by Theorem \ref{Theorem_apost_back} (see Fig. \ref{fig:back-analysis}). Below is an analysis similar to the fastest-finger case.
\begin{figure}[h]
\centering
    
\includegraphics[width = 0.32\linewidth]{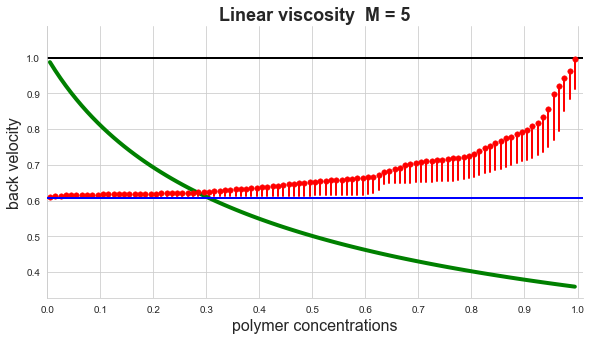} 
\includegraphics[width = 0.32\linewidth]{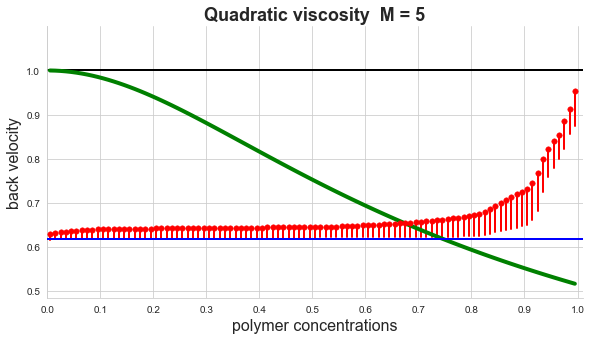} 
\includegraphics[width = 0.32\linewidth]{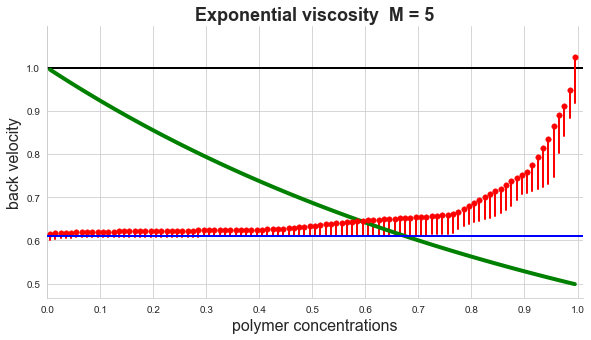}

\includegraphics[width = 0.32\linewidth]{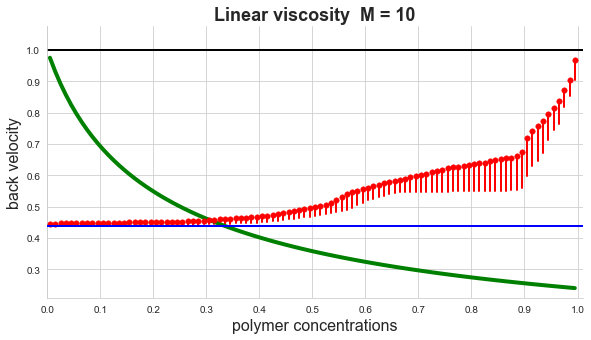} 
\includegraphics[width = 0.32\linewidth]{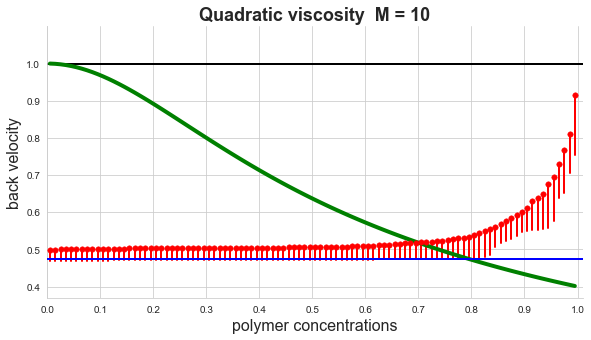} 
\includegraphics[width = 0.32\linewidth]{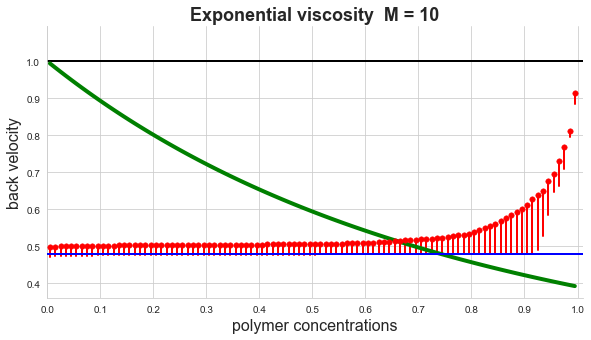}

\includegraphics[width = 0.32\linewidth]{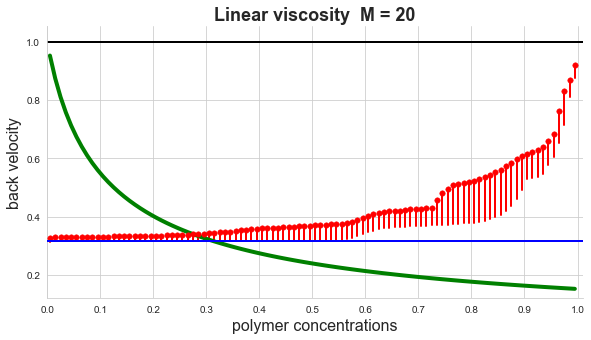} 
\includegraphics[width = 0.32\linewidth]{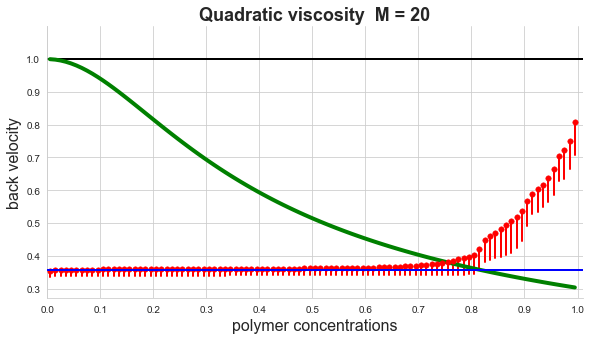} 
\includegraphics[width = 0.32\linewidth]{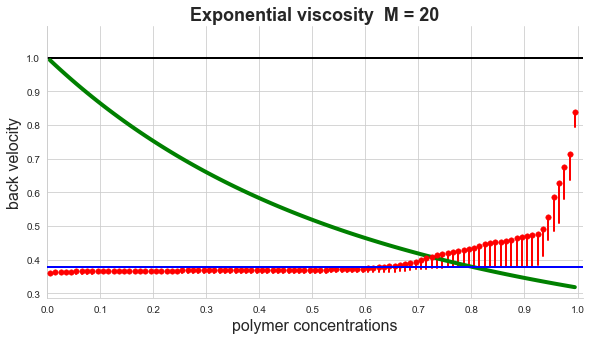}

\includegraphics[width = 0.32\linewidth]{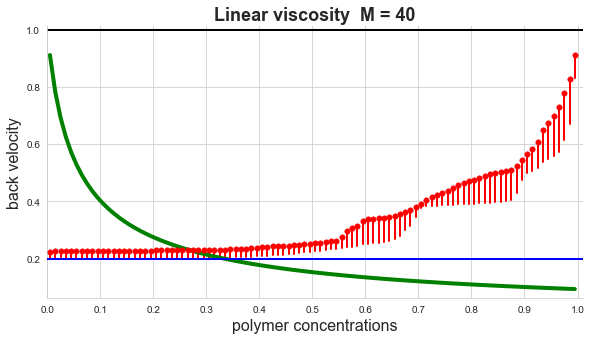} 
\includegraphics[width = 0.32\linewidth]{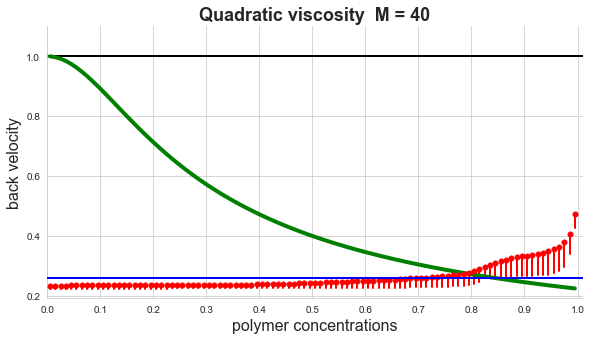} 
\includegraphics[width = 0.32\linewidth]{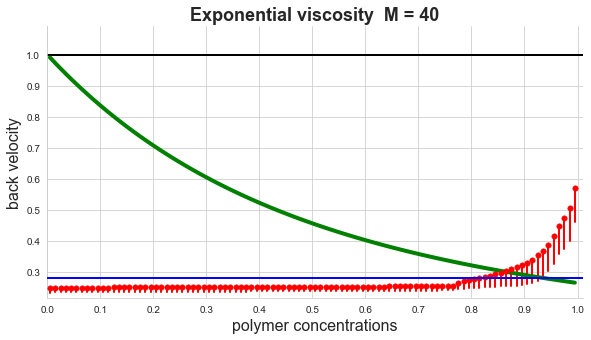}

\caption{Comparison of the speed of the back front with estimates by \eqref{TFE-estimate-mod-1}. The green curve corresponds to \eqref{TFE-estimate-mod-1}, filled dots correspond to \eqref{eq:fronts-regression}, lower point of the vertical strip corresponds to \eqref{eq:fronts-min}. The vertical strip shows the difference between  \eqref{eq:fronts-regression} and \eqref{eq:fronts-min}. The results are shown for linear (left column), quadratic (middle column), and exponential (right column) viscosity with contrasts (from top to bottom) $M = 5, 10, 20, 40$.} \label{fig:back-analysis}
\end{figure}
Similarly to the forward front we distinguish between three different situations:
\begin{itemize}
    \item[\textbf{A:}] The velocity function is approximately constant for $C<c_{*,0}$, has a significant increase for $C > c_{*,0}$ and fulfilled the estimates \eqref{eq-rem3}.
    \item[\textbf{B:}] The velocity function is approximately constant for $C<c_1$, where $c_1>c_{*,0}$ and fulfilled the estimates~\eqref{eq-rem3}.
    \item[\textbf{C:}] The velocity function is not approximately constant for $C<c_{*,0}$ or not satisfy estimates \eqref{eq-rem3}.    
\end{itemize}

Due to the significant nonlinearity of the back front for high values of $C$ we will not analyze Fig. \ref{fig:back-analysis} for velocities $v^b(C)$ defined by \eqref{eq:fronts-regression} and consider only velocities $v^b_{\min}(C)$. In Table  \ref{tab:back-summ-min} we summarize the results of the simulations for velocities $v^b_{\min}(C)$.
\begin{table}[h]
    \centering
    \begin{tabular}{c|c|c|c}
       & linear & quadratic & exponential  \\ \hline
M = 5  & A, B & A & A, B  \\
M = 10 & A & A & B \\
M = 20 & B & A, C & C\\
M = 40 & A & C & C
\end{tabular}
    \caption{Role of intermediate concentration for minimal velocity on the back front}
    \label{tab:back-summ-min}
\end{table}
In most of the cases classification could be made clearly. For linear and exponential viscosity with $M=5$ for $C>c_{*,0}$ we observe a very mild increase of the speed of the back front, which becomes much more significant for higher values of $C$, we label those cases \textbf{A, B}. For quadratic viscosity with $M=20$ the speed is not constant for $C < c_{*,0}$, at the same time it is constant for $C < c_{*,0}-0.01$, which is just one discrete step of $C$ of the graphics, we label this case as \textbf{A, C}. 

Note that for a small contrast of viscosity we clearly see cases \textbf{A} and \textbf{B}. For quadratic and exponential viscosity with high values of $M$ we see that the velocity of the back front starts to increase before intersection with the TFE estimates \eqref{TFE-estimate-mod-1}, which is not in agreement with Theorem \ref{Theorem_apost_back}. At the same time in those cases for $C > c_{*,0}$ locations of the front already experience significant jumps, see  Fig. \ref{fig:back} for quadratic viscosity with $M=40$ ($c_{*,0}/c_{\max} \sim 0.845$) and exponential viscosity with $M=20$ ($c_{*,0}/c_{\max} \sim 0.925$). 

Summarising the analysis of the back front, we do not observe a clear contradiction with Theorem \ref{Theorem_apost_back} and estimates \eqref{TFE-estimate-mod-1}. At the same time, clear confirmation of the estimates \eqref{TFE-estimate-mod-1} is observed in only 4 cases of 14. Therefore, we cannot confirm or disprove estimates \eqref{TFE-estimate-mod-1} for the back front.  

\section{Conclusions}
\label{conclusion}

We implemented a numerical scheme for viscous miscible displacement in homogeneous porous media for various viscosity functions and contrasts. Our choice of the model is motivated by its potential application in simulating post-flush scenarios following polymer flooding in oil reservoirs.
We study the evolution of the mixing zone and especially the speeds of its forefront and  rear front. Based on previous studies of correlations between those speeds and concentration levels, we suggest that two mechanisms play an important role in such relations: 
\begin{itemize}
    \item the diffusive slowdown according to the transverse flow equilibrium model;
    \item the development of an intermediate concentration within viscous fingers.
\end{itemize}
We propose improved estimates of the speeds of the two
edges of the mixing zone. The new estimates align well with numerical simulations of the speed of the forefront, but less so with the rear front, possibly due to significant non-linearity in the evolution of the rear front position over time. 
Note that the improved estimates depend on the values of intermediate concentrations. Further research is required to either estimate these values or determine them from experiments. 

\section*{Acknowledgements}

The authors are extremely grateful to the anonymous reviewers for their constructive suggestions and comments, which allowed to improve exposition of the manuscript. 

The research of Fedor Bakharev, Aleksandr Enin, Sergey Matveenko, Nikita Rastegaev, and Dmitry Pavlov was performed at
the Saint-Petersburg Leonhard Euler International Mathematical Institute and supported
by the Ministry of Science and Higher Education of the Russian Federation
(agreement no. 075-15-2022-287).

The research of Sergey Matveenko and Dmitry Pavlov was supported by ``Native towns'', a social investment program of PJSC ``Gazprom Neft''.

The research of Sergey Tikhomirov and Yulia Petrova were supported by Projeto Paz and Coordenação de Aperfeiçoamento de Pessoal de Nível Superior - Brasil (CAPES) - Finance Code 001. Sergey Tikhomirov is additionally supported by FAPERJ APQ1 E-26/210.702/2024 (295291) and CNPq grant 404123/2023-6. Yulia Petrova is additionally supported by FAPERJ APQ1 E-26/210.700/2024 (295287) and CNPq grant 406460/2023-0.

Figures \ref{fig:onefinger}, \ref{fig:borderfinger}, \ref{fig:Lin20-3levels},
\ref{fig:Exp5-3levels}, \ref{fig:quad10-2dim}  were obtained with Paraview~\cite{ParaView},
an open-source data visualization tool.

Figures \ref{fig:conjectureillustrated}, \ref{fig:front-analysis}, 
\ref{fig:quad10-mod}, \ref{fig:back-analysis}, \ref{fig:finer-grid}  were obtained with Seaborn~\cite{Waskom2021},
an open-source data visualization tool.

\section*{Declaration of Interests} The authors report no conflict of interest.

\section*{Appendix A}

\begin{proof}[Proof of Theorem~\ref{Theorem_apost}]
The main step in the proof (Lemma~\ref{Lemma_apost} below) is the construction of a partial lower bound for the solution $c(t,x,y)$ of \eqref{03-TFE-1}, \eqref{03-TFE} in the form of a traveling wave solution $c_{lb,\delta}(t,x)$ of the equation
\begin{equation}
\label{appendix-Lower_bound_equation}
\dfrac{\partial c_{lb,\delta}}{\partial t} + \dfrac{m(c_{lb,\delta})}{m(c_{\max})}\dfrac{\partial c_{lb,\delta}}{\partial x} = D \dfrac{\partial^2 c_{lb,\delta}}{\partial x^2}
\end{equation}
connecting levels $(c^*-\delta)$ and $(c_{\max}-\delta)$ for any small $\delta>0$. Recall that a function $g(t, x): \mathbb{R}_+ \times\mathbb{R} \to \mathbb{R}$, 
is a \textit{traveling wave} with speed $v\in\mathbb{R}$ connecting states $g_-\in\mathbb{R}$ and $g_+\in\mathbb{R}$, if it has the form $g(t,x)=w(x-vt)$, where $w:\mathbb{R}\to\mathbb{R}$ 
is a continuous function, which satisfies $w(-\infty)=g_-$ and $w(+\infty)=g_+$.
\begin{remark}
    To show the existence of a traveling wave solution $c_{lb,\delta}(t,x)$ 
    for equation~\eqref{appendix-Lower_bound_equation} connecting states $(c^*-\delta)$ and $(c_{\max}-\delta)$, we just
    rewrite the equation~\eqref{appendix-Lower_bound_equation} as
    \begin{align}
    \label{appendix-Lower_bound_equation1}
        \dfrac{\partial c_{lb,\delta}}{\partial t} + \dfrac{\partial F(c_{lb,\delta})}{\partial x} = D \dfrac{\partial^2 c_{lb,\delta}}{\partial x^2},\qquad F(c)=\frac{1}{m(c_{\max})}\int\limits_{0}^c m(c')\,dc',
    \end{align}
    and notice that the function $F(c)$ is concave ($F''(c)=m'(c)<0$) and $c^*-\delta<c_{\max}-\delta$. The speed of the traveling wave is defined by the Rankine-Hugoniot condition~\cite[Sect.~3.4.1]{evans2022pdes}:
    \begin{align}
    \label{RH}
        v=\frac{F(c_{\max}-\delta)-F(c^*-\delta)}{(c_{\max}-\delta) - (c^*-\delta)}
        =\frac{\overline{m}(c^*-\delta,c_{\max}-\delta)}{m(c_{\max})}.
    \end{align}
\end{remark}

The idea of comparison of the original solution  with a traveling wave solution of a simpler one-dimensional equation~\eqref{appendix-Lower_bound_equation} was proposed in~\cite{Otto2006, Yortsos} for the case $c^* = 0$. We generalize this idea for $c^*>0$ 
and formalize it in the following Lemma.
\begin{lemma}
\label{Lemma_apost}
Under the assumptions of Theorem \ref{Theorem_apost}, let $c_{lb, \delta}(t,x)$ be a traveling wave solution of \eqref{appendix-Lower_bound_equation} with limits $(c^*-\delta)$ as $x\to -\infty$ and $(c_{\max}-\delta)$ as $x\to +\infty$, such that
\begin{equation}
\label{eq-travelling-wave-below}
c(0,x,y) > c_{lb,\delta}(0,x), \qquad \forall x \geqslant x^f(0, c^*),\, y\in[0,H],
\end{equation}
and also
\begin{align}
\label{init-inequality}
c_{lb,\delta}(0, l_1) < c^*.
\end{align}
Then
\begin{align}
\label{comparison}
    c(t,x,y) > c_{lb,\delta}(t,x), \qquad \forall t>0, \, x \geqslant x^f(t, c^*),\,  y\in[0,H].
\end{align}
\end{lemma}

\begin{proof}
Let us prove the inequality~\eqref{comparison} for two cases separately: (a) on the boundary for $x=x^f(t,c^*)$; and (b) inside the domain for $x>x^f(t,c^*)$.

(a) Note that due to the Rankine-Hugoniot condition the wave $c_{lb, \delta}$ travels with speed, defined in~\eqref{RH}, which we estimate
\[
\dfrac{\overline{m}(c^*-\delta, c_{\max}-\delta)}{m(c_{\max})} > 
\dfrac{\overline{m}(c^*, c_{\max})}{m(c_{\max})}
\]
due to monotonicity of the mobility function $m$. Thus, using~\eqref{x-front},~\eqref{a_post_est} and~\eqref{init-inequality}, 
we obtain
\begin{align*}
c(t, x^f(t, c^*), y) &\geqslant c^* > c_{lb,\delta}(0,l_1) 
= c_{lb,\delta}\Big(t,\frac{\overline{m}(c^*-\delta, c_{\max}-\delta)}{c_{\max}}\cdot t + l_1\Big) > c_{lb,\delta}(t, x^f(t, c^*)).
\end{align*}
Therefore, on the boundary we have
\[
c(t,x,y) > c_{lb,\delta}(t,x), \qquad x = x^f(t, c^*), t>0, y \in [0, H].
\]
Thus, if the inequality~\eqref{comparison} breaks, it must be inside the area $x > x^f(t, c^*)$.

(b) Consider the function $\vartheta(t,x,y):=c(t,x,y)-c_{lb,\delta}(t,x)$. We aim to demonstrate that $\vartheta>0$ for all $t>0, x > x^f(t, c^*),  y\in[0,H]$. Proof by contradiction. 
Let $t^*>0$ be the smallest time for which the inequality~\eqref{comparison} breaks, and $(x^*,y^*)$ be the space point where it happens. This gives us several estimates. First,
\begin{equation}
\label{01-B5} 
\vartheta(t^*,x^*,y^*)=0.
\end{equation}
Second,
\begin{equation}
\label{01-B6} 
\partial_t \vartheta(t^*,x^*,y^*) \leqslant 0,
\end{equation}
since for all $t<t^*$ we have $\vartheta(t,x^*,y^*)>0$. Third, if $y^*\in(0, H)$, then
\begin{equation}
\label{01-B7} 
\partial_x \vartheta(t^*,x^*,y^*)=\partial_y \vartheta(t^*,x^*,y^*)= 0 \quad \mbox{and} \quad \Delta \vartheta(t^*,x^*,y^*)\geqslant 0,
\end{equation}
since for all $x \geqslant x^f(t, c^*),  y\in[0,H]$ we have $\vartheta(t^*,x,y)\geqslant 0=\vartheta(t^*,x^*,y^*)$. And if $y^* = 0 \text{ or } H$, then
\begin{equation}
\label{01-B8} 
q^y(t^*, x^*, y^*) = 0
\end{equation}
due to no-flow boundary condition.

On the other hand, combining \eqref{01-Peaceman_model} and \eqref{appendix-Lower_bound_equation}, we obtain
\begin{align*}
\partial_t\vartheta(t^*,x^*,y^*)+q^x\partial_x\vartheta(t^*,x^*,y^*)&+q^y\partial_y\vartheta(t^*,x^*,y^*)-D \Delta \vartheta(t^*,x^*,y^*) \\
{} &= \left(\frac{m(c_{lb,\delta}(x^*,t^*))}{m(c_{\max})} - q^x(t^*,x^*,y^*)\right)\partial_x c_{lb,\delta}(t^*,x^*).
\end{align*}
The left-hand side is non-positive due to the estimates~\eqref{01-B5}--\eqref{01-B8}
and the right-hand side is strictly positive due to the monotonicity of the traveling wave solution (it strictly increases, i.e. $\partial_x c_{lb,\delta}>0$) and the trivial estimate
\[
q^x = \dfrac{m(c)}{\displaystyle \frac{1}{H}\int_0^H\!\!\! m(c)\,dy} < \dfrac{m(c)}{m(c_{\max})}.
\]
We arrive at a contradiction, thus proving Lemma~\ref{Lemma_apost}.
\end{proof}

Let us finish the proof of Theorem~\ref{Theorem_apost} using Lemma~\ref{Lemma_apost}. 

For any fixed $\delta > 0$ it is possible to construct a traveling wave satisfying the assumptions of Lemma \ref{Lemma_apost}. We just need to shift its initial state sufficiently to the right, thereby satisfying the relation \eqref{eq-travelling-wave-below}. As the result, $x^f(t,\widetilde{c})$ can never outrun the resulting traveling wave. Therefore, for any $\widetilde{c} \in (c^*, c_{\max})$ there exists a constant $l_2 > 0$, which depends on $\delta$, such that 
\begin{align}
\label{front-estimate-1}
x^f(t, \widetilde{c}) \leqslant \dfrac{\overline{m}(c^*-\delta, c_{\max}-\delta)}{m(c_{\max})} \cdot t + l_2.
\end{align}
As $\delta$ can be any positive small number, we immediately get the inequality~\eqref{a_post_est1}. Theorem~\ref{Theorem_apost} is proven.
\end{proof}

\begin{corollary}
\label{corollary-thm1}
Assume that $c(0,x,y)=c_{\max}$ for all $x\geqslant x_0$ and $y\in[0,H]$. Using the explicit form of the traveling wave solution for the equation~\eqref{Lower_bound_equation}, we can formulate the main conclusion of Theorem~\ref{Theorem_apost} in an alternative form: 
for any $\widetilde{c} \in (c^*, c_{\max})$ there exists a function $l_2(t) > 0$ such that $l_2(t) = O(\ln(t))$ as $t\to+\infty$ and 
\begin{equation}
\label{a_post_est2}
x^f(t, \widetilde{c}) \leqslant \dfrac{\overline{m}(c^*, c_{\max})}{m(c_{\max})} \cdot t + l_2(t).
\end{equation}

Similarly, if we postulate that $c(0,x,y)=0$ for all $x\leqslant 0$ and $y\in[0,H]$, then the  conclusion of Theorem~\ref{Theorem_apost_back} can be formulated in an alternative form:  
for any $\widetilde{c} \in (0, c_*)$ there exists a function $l_2(t) > 0$ such that $l_2(t) = O(\ln(t))$ as $t\to+\infty$ and 
\begin{equation}
\label{a_post_est_back2}
x^b(t, \widetilde{c}) \geqslant \dfrac{\overline{m}(0, c_*)}{m(0)} \cdot t - l_2(t).
\end{equation}
\end{corollary}
\begin{proof}
We prove \eqref{a_post_est2} here, the proof of \eqref{a_post_est_back2} is similar. Let us write more accurately the dependence on $\delta$ in the estimate~\eqref{front-estimate-1}. In particular, let $l_2=l_2(\delta)$. We further estimate 
\[
\dfrac{\overline{m}(c^*-\delta, c_{\max}-\delta)}{m(c_{\max})} \leqslant \dfrac{\overline{m}(c^*, c_{\max})}{m(c_{\max})} + C_{\overline{m}}\delta
\]
for some $C_{\overline{m}}>0$, arriving at
\begin{align}
\label{appendix-estimate}
x^f(t, \widetilde{c}) \leqslant \dfrac{\overline{m}(c^*, c_{\max})}{m(c_{\max})} \cdot t + l_2(\delta) + C_{\overline{m}}\delta t.
\end{align}
To make further progress, we need to estimate the growth of $l_2(\delta)$ as $\delta\to 0$. Note that $l_2(\delta)$ is bounded from above by the initial position $x_{\widetilde{c}}$ of the level set $\widetilde{c}$ of the traveling wave $c_{lb,\delta}$ (by definition $c_{lb,\delta}(0,x_{\widetilde{c}})=\widetilde{c}$). To finish the proof, it is enough to show that
\begin{align}
\label{logarithmic-bound}
    l_2(\delta)\leqslant x_{\widetilde{c}}=O(|\ln(\delta)|), \quad \delta \to 0.
\end{align}
Indeed, by choosing $\delta = 1/t$ in Eq.~\eqref{appendix-estimate}, we obtain the necessary estimate~\eqref{a_post_est2}.

\medskip
In order to show the logarithmic upper bound~\eqref{logarithmic-bound}, let us be more precise in how we fix the initial shift of the traveling wave.  For any $\delta>0$ we specify the position $x_{c^*}$ of the level set $c^*$, that is $c_{lb,\delta}(0,x_{c^*})=c^*$, such that the assumptions~\eqref{eq-travelling-wave-below} and~\eqref{init-inequality} are satisfied. In particular, to guarantee this we can assign $x_{c^*} = \max(x_0, l_1+\delta)$ 
for all $\delta > 0$.
The  condition~\eqref{eq-travelling-wave-below} is satisfied for $x\in (x^f(0,c^*),x_{c^*})$ due to inequality $c(0,x,y)>c^*>c_{lb,\delta}(0,x)$; and for $x\geqslant x_{c^*}$ due to trivial estimate $c(0,x,y)=c_{\max}>c_{lb,\delta}(0,x)$. The condition~\eqref{init-inequality} is satisfied due to $l_1<x_{c^*}$.
 
The following formula is valid
\begin{align}
\label{formula-TW}
x_{\widetilde{c}} = x_{c^*}+  \int\limits_{c^*}^{\widetilde{c}} \dfrac{D \, dr}{(r - c^* + \delta)\left[\frac{\overline{m}(c^*-\delta, r)}{m(c_{\max})} - 
v \right]},
\end{align}
 and provides the logarithmic bound for $x_{\widetilde{c}}$ as $\delta\to0$ due to the  linear growth of the denominator in the integrand of~\eqref{formula-TW} near $r=c^*$.  
 Indeed, the second multiplicative term in the denominator is asymptotically constant near $r = c^*$ as $\delta\to0$
 \begin{align*}
     \frac{\overline{m}(c^*-\delta, r)}{m(c_{\max})} - v = \frac{\overline{m}(c^*-\delta, r)}{m(c_{\max})} - \frac{\overline{m}(c^*-\delta,c_{\max}-\delta)}{m(c_{\max})}=\frac{m(c^*)-\overline{m}(c^*,c_{\max})}{m(c_{\max})}(1+o(1)).
 \end{align*}

The derivation of~\eqref{formula-TW} is standard, but for completeness of the proof we present it. Putting the traveling wave ansatz $c_{lb,\delta}(t,x)=w(x-vt)$, $\xi=x-vt$, into equation~\eqref{appendix-Lower_bound_equation1}, we get
\begin{align*}
    -
    v w'(\xi) + (F(w(\xi)))'=Dw''(\xi).
\end{align*}
Integrating from $-\infty$ to $\xi$, and using the fact that $w'(-\infty)=0$, we obtain
\begin{align*}
    [-
    v w(\xi)+F(w(\xi))]\biggr\rvert_{-\infty}^\xi
    =Dw'(\xi).
\end{align*}
As $w=w(\xi)$ is a strictly increasing function, we can define the inverse function $\xi=\xi(w)$, and using $w(-\infty)=c^*-\delta$, we get
\begin{align*}
    \xi'(w)=\frac{D}{-v(w-c^*+\delta)+F(w)-F(c^*-\delta)}.
\end{align*}
Integrating now in $w$ from $c^*$ to $\widetilde{c}$ and using the definitions of $F$ and $\overline{m}$, we get formula
\begin{align*}
\xi(\widetilde{c}) = \xi(c^*)+  \int\limits_{c^*}^{\widetilde{c}} \dfrac{D \, dr}{(r - c^* + \delta)\left[\frac{\overline{m}(c^*-\delta, r)}{m(c_{\max})} - 
v \right]}.
\end{align*}
Putting $t=0$ (that is $\xi=x$), we finally arrive at~\eqref{formula-TW}. The estimate~\eqref{a_post_est2} in Corollary~\ref{corollary-thm1} is proven.
\end{proof}

\section*{Appendix B}
Two additional simulations were performed on a finer grid ($5120\times 2880$) to further test the 
estimates (\ref{TFE-estimate-mod-estimates}) with smaller (numerical) diffusion and
a hence higher P\'eclet number. The chosen two setups were: linear viscosity with $M=5$
and exponential viscosity with $M=5$. We do not claim that we made systematic investigation of simulations for this value of P\'eclet number. The main purpose of those simulations was to check if the speed of the mixing zone and intermediate concentration can seriously change with growth of P\'eclet number. 

\begin{figure}[h]
\centering
\includegraphics[width = 0.43\linewidth]{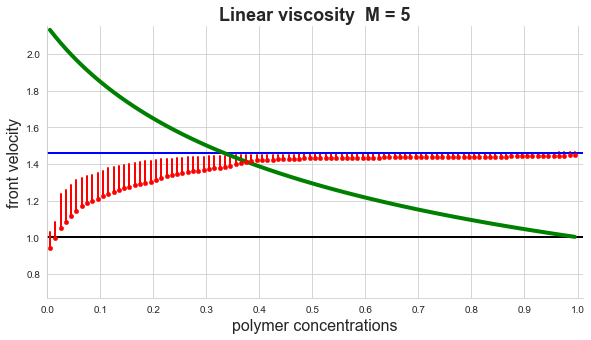}
\includegraphics[width = 0.43\linewidth]{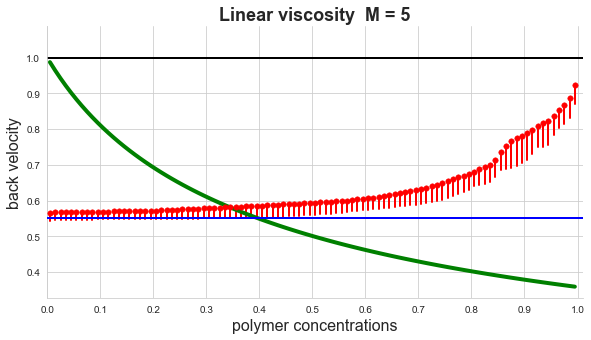}
\includegraphics[width = 0.43\linewidth]{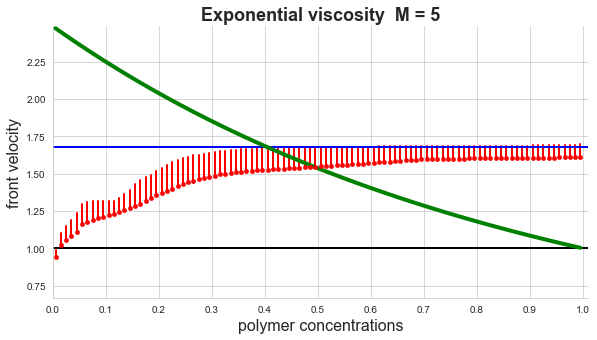}
\includegraphics[width = 0.43\linewidth]{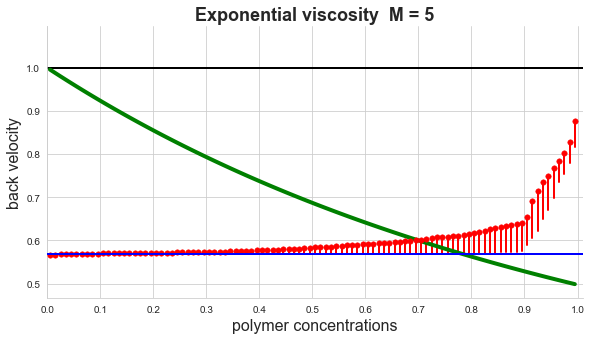}
    \caption{Comparison of the speed of forefront (left) and back front (right) with \eqref{TFE-estimate-mod-1}, for two simulations made on a $5120\times 2880$ grid.}
    \label{fig:finer-grid}
\end{figure}

Applying analysis similar to Sect.~\ref{sec:forward-front-12},~\ref{sec:back-front-12} (see Fig.~\ref{fig:finer-grid}) we can conclude the following classification of the fronts behavior for maximal and minimal velocities, see Table \ref{tab:finer-ABC}. All four cases are classified as \textbf{A} or \textbf{B}, which does not contradict our estimates \eqref{TFE-estimate-mod-estimates}. Values of intermediate concentrations $c_0^*/c_{\max}$, $c_{*, 0}/c_{\max}$ for grids $2560 \times 1440$ and $5120\times 2880$ are shown in Table \ref{tab:compare-grids}. It is important to note that the value of intermediate concentration for fastest finger and back front are far from 0 and 1 respectively. For fastest finger it actually increased, which is potentially an indicator of high variability of the value of intermediate concentration. Note that we do not claim in the paper any estimates on intermediate concentration or monotonicity of its dependence on P\'eclet number, but relation between the speed and intermediate concentration: equations~\eqref{TFE-estimate-mod-1}, Tables~\ref{tab:summ-max},~\ref{tab:back-summ-min},~\ref{tab:finer-ABC}.

\begin{table}[h]
\centering
    \begin{tabular}{r|c|c}
       & front & back   \\ \hline
linear, M = 5  & A & A   \\
exponential, M = 5 & A, B & B  \\
\end{tabular}
    \caption{Role of intermediate concentration for minimal/maximal velocity on a finer grid}
    \label{tab:finer-ABC}
\medskip    
\begin{tabular}{r|c|c}
       & linear, M = 5 &  exponential, M = 5  \\ \hline
$2560 \times 1440$  & $c_0^*/c_{\max} = 0.195$, $c_{*, 0}/c_{\max} = 0.305$ & $c_0^*/c_{\max} = 0.325$, $c_{*, 0}/c_{\max} = 0.675$ \\
$5120\times 2880$  & $c_0^*/c_{\max} = 0.335$, $c_{*, 0}/c_{\max} = 0.395$ & $c_0^*/c_{\max} = 0.405$, $c_{*, 0}/c_{\max} = 0.785$  \\
\end{tabular}
    \caption{Values of intermediate concentration on original and finer grid.}
    \label{tab:compare-grids}
\end{table}

\newpage 
\bibliography{references.bib}

\begin{thebibliography}{10}
\expandafter\ifx\csname url\endcsname\relax
  \def\url#1{\texttt{#1}}\fi
\expandafter\ifx\csname urlprefix\endcsname\relax\def\urlprefix{URL }\fi
\expandafter\ifx\csname href\endcsname\relax
  \def\href#1#2{#2} \def\path#1{#1}\fi

\bibitem{bedrikovetsky2013}
P.~Bedrikovetsky, Mathematical theory of oil and gas recovery: with applications to ex-USSR oil and gas fields, Vol.~4, Springer Science \& Business Media, 2013.

\bibitem{SaffmanTaylor1958}
P.~G. Saffman, G.~I. Taylor, The penetration of a fluid into a porous medium or {Hele-Shaw} cell containing a more viscous liquid, Proceedings of the Royal Society of London. Series A. Mathematical and Physical Sciences 245~(1242) (1958) 312--329.

\bibitem{Wooding1969}
R.~A. Wooding, Growth of fingers at an unstable diffusing interface in a porous medium or {Hele-Shaw} cell, Journal of Fluid Mechanics 39~(3) (1969) 477--495.

\bibitem{Homsy1987}
G.~M. Homsy, Viscous fingering in porous media, Annual Review of Fluid Mechanics 19~(1) (1987) 271--311.

\bibitem{Tanveer2000}
S.~Tanveer, Surprises in viscous fingering, Journal of Fluid Mechanics 409 (2000) 273--308.

\bibitem{YYS2002}
Z.~M. Yang, Y.~C. Yortsos, D.~Salin, Asymptotic regimes in unstable miscible displacements in random porous media, Advances in Water Resources 25~(8-12) (2002) 885--898.

\bibitem{Nijjer2018}
J.~Nijjer, D.~Hewitt, J.~Neufeld, The dynamics of miscible viscous fingering from onset to shutdown, Journal of Fluid Mechanics 837 (2018) 520--545.
\newblock \href {https://doi.org/10.1017/jfm.2017.829} {\path{doi:10.1017/jfm.2017.829}}.

\bibitem{scovazzi2017}
G.~Scovazzi, M.~F. Wheeler, A.~Mikeli{\'c}, S.~Lee, Analytical and variational numerical methods for unstable miscible displacement flows in porous media, Journal of Computational Physics 335 (2017) 444--496.

\bibitem{ChuokeMeurs1959}
R.~Chuoke, P.~Van~Meurs, C.~van~der Poel, {The} instability of slow, immiscible, viscous liquid-liquid displacements in permeable media, Transactions of the AIME 216~(01) (1959) 188--194.

\bibitem{Outmans1962}
H.~Outmans, Nonlinear theory for frontal stability and viscous fingering in porous media, Society of Petroleum Engineers Journal 2~(02) (1962) 165--176.
\newblock \href {https://doi.org/10.2118/183-PA} {\path{doi:10.2118/183-PA}}.

\bibitem{Perrine1961-I}
R.~Perrine, Stability theory and its use to optimize solvent recovery of oil, Society of Petroleum Engineers Journal 1~(01) (1961) 9--16.
\newblock \href {https://doi.org/10.2118/1508-G} {\path{doi:10.2118/1508-G}}.

\bibitem{Perrine1961-II}
R.~Perrine, The development of stability theory for miscible liquid-liquid displacement, Society of Petroleum Engineers Journal 1~(01) (1961) 17--25.
\newblock \href {https://doi.org/10.2118/1509-G} {\path{doi:10.2118/1509-G}}.

\bibitem{Claridge}
E.~Claridge, A method for designing graded viscosity banks, Society of Petroleum Engineers Journal 18~(05) (1978) 315--324.
\newblock \href {https://doi.org/10.2118/6848-PA} {\path{doi:10.2118/6848-PA}}.

\bibitem{Lake}
L.~W. Lake, R.~T. Johns, B.~Rossen, G.~Pope, Fundamentals of enhanced oil recovery, 2014.

\bibitem{green1998eor}
D.~W. Green, G.~P. Willhite, Enhanced oil recovery (1998).

\bibitem{Peaceman1962}
D.~Peaceman, H.~Rachford, Numerical calculation of multidimensional miscible displacement, Society of Petroleum Engineers Journal 2~(04) (1962) 327--339.
\newblock \href {https://doi.org/10.2118/471-PA} {\path{doi:10.2118/471-PA}}.

\bibitem{samanta2011}
A.~Samanta, K.~Ojha, A.~Sarkar, A.~Mandal, Surfactant and surfactant-polymer flooding for enhanced oil recovery, Advances in Petroleum Exploration and Development 2~(1) (2011) 13--18.

\bibitem{GVB}
F.~Bakharev, A.~Enin, K.~Kalinin, Y.~Petrova, N.~Rastegaev, S.~Tikhomirov, Optimal polymer slugs injection profiles, Journal of Computational and Applied Mathematics 425 (2023) 115042.

\bibitem{tikhomirov2021spe}
S.~Tikhomirov, F.~Bakharev, A.~Groman, A.~Kalyuzhnyuk, Y.~Petrova, A.~Enin, K.~Kalinin, N.~Rastegaev, Calculation of graded viscosity banks profile on the rear end of the polymer slug, in: SPE Russian Petroleum Technology Conference, SPE, 2021.
\newblock \href {https://doi.org/10.2118/206426-MS} {\path{doi:10.2118/206426-MS}}.

\bibitem{Linear}
F.~Bakharev, A.~Enin, A.~Groman, A.~Kalyuzhnyuk, S.~Matveenko, Y.~Petrova, I.~Starkov, S.~Tikhomirov, Velocity of viscous fingers in miscible displacement: Comparison with analytical models, Journal of Computational and Applied Mathematics 402 (2022) 113808.
\newblock \href {https://doi.org/10.1016/j.cam.2021.113808} {\path{doi:10.1016/j.cam.2021.113808}}.

\bibitem{Koval}
{Koval, E.J.}, {A} method for predicting the performance of unstable miscible displacement in heterogeneous media, Society of Petroleum Engineers Journal 3~(02) (1963) 145--154.
\newblock \href {https://doi.org/10.2118/450-PA} {\path{doi:10.2118/450-PA}}.

\bibitem{Booth}
R.~Booth, {On} the growth of the mixing zone in miscible viscous fingering, Journal of Fluid Mechanics 655 (2010) 527--539.
\newblock \href {https://doi.org/10.1017/S0022112010001734} {\path{doi:10.1017/S0022112010001734}}.

\bibitem{TL}
M.~Todd, W.~Longstaff, The development, testing, and application of a numerical simulator for predicting miscible flood performance, Journal of Petroleum Technology 24~(07) (1972) 874--882.
\newblock \href {https://doi.org/10.2118/3484-PA} {\path{doi:10.2118/3484-PA}}.

\bibitem{Yortsos}
Y.~Yortsos, D.~Salin, {On} the selection principle for viscous fingering in porous media, Journal of Fluid Mechanics 557 (2006) 225.
\newblock \href {https://doi.org/10.1017/S0022112006009761} {\path{doi:10.1017/S0022112006009761}}.

\bibitem{Otto2005}
G.~Menon, F.~Otto, Dynamic scaling in miscible viscous fingering, Communications in Mathematical Physics 257~(2) (2005) 303--317.
\newblock \href {https://doi.org/10.1007/s00220-004-1264-7} {\path{doi:10.1007/s00220-004-1264-7}}.

\bibitem{Otto2006}
G.~Menon, F.~Otto, Diffusive slowdown in miscible viscous fingering, Communications in Mathematical Sciences 4~(1) (2006) 267--273.
\newblock \href {https://doi.org/10.4310/CMS.2006.v4.n1.a11} {\path{doi:10.4310/CMS.2006.v4.n1.a11}}.

\bibitem{PTY}
Y.~Petrova, S.~Tikhomirov, Y.~Efendiev, Propagating terrace in a two-tubes model of gravitational fingering (2024).
\newblock \href {http://arxiv.org/abs/2401.05981} {\path{arXiv:2401.05981}}.

\bibitem{CGS}
I.~A. Starkov, D.~A. Pavlov, S.~B. Tikhomirov, F.~L. Bakharev, The non-monotonicity of growth rate of viscous fingers in heterogeneous porous media, Computational Geosciences 27~(5) (2023) 783–792.
\newblock \href {https://doi.org/10.1007/s10596-023-10240-3} {\path{doi:10.1007/s10596-023-10240-3}}.

\bibitem{chen2001}
C.-Y. Chen, S.-W. Wang, Miscible displacement of a layer with finite width in porous media, International Journal of Numerical Methods for Heat \& Fluid Flow (2001).

\bibitem{dewit2005}
A.~De~Wit, Y.~Bertho, M.~Martin, Viscous fingering of miscible slices, Physics of Fluids 17~(5) (2005) 054114.

\bibitem{mishra2008}
M.~Mishra, M.~Martin, A.~De~Wit, Differences in miscible viscous fingering of finite width slices with positive or negative log-mobility ratio, Physical Review E 78~(6) (2008) 066306.

\bibitem{pramanik2016}
S.~Pramanik, M.~Mishra, Coupled effect of viscosity and density gradients on fingering instabilities of a miscible slice in porous media, Physics of Fluids 28~(8) (2016) 084104.

\bibitem{sharma2021}
V.~Sharma, H.~B. Othman, Y.~Nagatsu, M.~Mishra, Viscous fingering of miscible annular ring, Journal of Fluid Mechanics 916 (2021).

\bibitem{Laz}
D.~E. Apushkinskaya, G.~G. Lazareva, V.~A. Okishev, Influence of numerical diffusion on the growth rate of viscous fingers in the numerical implementation of the {Peaceman} model by the finite volume method, Sovrem. Mat. Fundam. Napravl. 68 (2022) 553--563.
\newblock \href {https://doi.org/10.22363/2413-3639-2022-68-4-553-563} {\path{doi:10.22363/2413-3639-2022-68-4-553-563}}.

\bibitem{flemisch2011dumux}
B.~Flemisch, M.~Darcis, K.~Erbertseder, B.~Faigle, A.~Lauser, K.~Mosthaf, S.~M{\"u}thing, P.~Nuske, A.~Tatomir, M.~Wolff, et~al., {DuMux}: {DUNE} for multi-$\{$phase, component, scale, physics,…$\}$ flow and transport in porous media, Advances in Water Resources 34~(9) (2011) 1102--1112.

\bibitem{Kochetal2020Dumux}
T.~Koch, D.~Gläser, K.~Weishaupt, S.~Ackermann, M.~Beck, B.~Becker, S.~Burbulla, H.~Class, E.~Coltman, S.~Emmert, T.~Fetzer, C.~Grüninger, K.~Heck, J.~Hommel, T.~Kurz, M.~Lipp, F.~Mohammadi, S.~Scherrer, M.~Schneider, G.~Seitz, L.~Stadler, M.~Utz, F.~Weinhardt, B.~Flemisch, {DuMu}\textsuperscript{x} 3 --- an open-source simulator for solving flow and transport problems in porous media with a focus on model coupling, Computers \& Mathematics with Applications (2020).
\newblock \href {https://doi.org/10.1016/j.camwa.2020.02.012} {\path{doi:10.1016/j.camwa.2020.02.012}}.

\bibitem{ans-DUNE}
M.~Blatt, A.~Burchardt, A.~Dedner, C.~Engwer, J.~Fahlke, B.~Flemisch, C.~Gersbacher, C.~Gr\"{a}ser, F.~Gruber, C.~Gr\"{u}ninger, D.~Kempf, R.~Kl\"{o}fkorn, T.~Malkmus, S.~M\"{u}thing, M.~Nolte, M.~Piatkowski, O.~Sander, The {Distributed} and {Unified} {Numerics} {Environment}, version 2.4, Archive of Numerical Software 4~(100) (2016).
\newblock \href {https://doi.org/10.11588/ANS.2016.100.26526} {\path{doi:10.11588/ANS.2016.100.26526}}.

\bibitem{gmsh}
C.~Geuzaine, J.-F. Remacle, Gmsh: {A} {3-D} finite element mesh generator with built-in pre- and post-processing facilities, International Journal for Numerical Methods in Engineering 79~(11) (2009) 1309--1331.
\newblock \href {https://doi.org/10.1002/nme.2579} {\path{doi:10.1002/nme.2579}}.

\bibitem{ParaView}
J.~Ahrens, B.~Geveci, C.~Law, {ParaView}: An end-user tool for large data visualization, in: Visualization Handbook, Elesvier, 2005, {ISBN}~978-0123875822.

\bibitem{Waskom2021}
M.~L. Waskom, Seaborn: statistical data visualization, Journal of Open Source Software 6~(60) (2021) 3021.
\newblock \href {https://doi.org/10.21105/joss.03021} {\path{doi:10.21105/joss.03021}}.

\bibitem{evans2022pdes}
L.~C. Evans, Partial differential equations, Vol.~19, American Mathematical Society, 2022.

\end{thebibliography}

\end{document}